\documentclass[a4paper,onecolumn,accepted=2024-12-19]{quantumarticle}
\pdfoutput=1
\usepackage[utf8]{inputenc}
\usepackage[english]{babel}
\usepackage[T1]{fontenc}
\usepackage{amsmath}
\usepackage{hyperref}

\usepackage{tikz,cite}
\usepackage{lipsum}

\usepackage{graphicx}
\usepackage{xcolor}
\usepackage{amsmath,amsthm,amssymb}
\usepackage[margin=1in]{geometry}
\usepackage{hyperref}

\usepackage{multirow}

\newcommand{\F}{\mathbb{F}}

\def\ba{\begin{array}}
\def\ea{\end{array}}

\DeclareMathOperator*{\argmin}{arg\,min}

\def\0{{\bf 0}}
\def\a{{\bf a}}
\def\b{{\bf b}}

\def\t{{\bf t}}
\def\C{{\mathcal C}}
\def\P{{\mathcal P}}
\def\D{{\mathcal D}}
\def\F{{\tt F}}

\def\x{{\bf x}}
\def\y{{\bf y}}

\def\p{{\bf p}}
\def\q{{\bf q}}

\def\u{{\bf u}}
\def\v{{\bf v}}

\def\T{{\tt T}}

\newcommand{\ket}[1]{| #1 \rangle}
\newcommand{\bra}[1]{\langle #1|}

\newcommand{\be}{\begin{equation}}
\newcommand{\ee}{\end{equation}}
\newcommand{\bea}{\begin{eqnarray}}
\newcommand{\eea}{\end{eqnarray}}
\newcommand{\bes}{\begin{equation*}}
\newcommand{\ees}{\end{equation*}}
\newcommand{\beas}{\begin{eqnarray*}}
\newcommand{\eeas}{\end{eqnarray*}}

\makeatletter
\newtheorem*{rep@theorem}{\rep@title}
\newcommand{\newreptheorem}[2]{%
\newenvironment{rep#1}[1]{%
 \def\rep@title{#2 \ref{##1} (restated)}%
 \begin{rep@theorem}}%
 {\end{rep@theorem}}}
\makeatother

\allowdisplaybreaks

\newtheorem{thm}{Theorem}
\newtheorem*{thm*}{Theorem}

\newtheorem{lem}[thm]{Lemma}
\newtheorem*{lem*}{Lemma}

\newtheorem{prop}[thm]{Proposition}
\newtheorem{defn}[thm]{Definition}
\newtheorem{rem}[thm]{Remark}

\newreptheorem{thm}{Theorem}
\newreptheorem{lem}{Lemma}
\newreptheorem{cor}{Corollary}

\newtheorem*{prob*}{Problem}



\begin{document}

\title{Lower bounds for quantum-inspired classical algorithms via communication complexity }

\author{Nikhil S. Mande}
\affiliation{Department of Computer Science, University of Liverpool, Liverpool L69 3BX, UK}
\email{nikhil.mande@liverpool.ac.uk}

\author{Changpeng Shao}
\affiliation{Academy of Mathematics and Systems Science, Chinese Academy of Sciences, Beijing, 100190 China}
\email{changpeng.shao@amss.ac.cn}


\maketitle

\begin{abstract}
Quantum-inspired classical algorithms provide us with a new way to understand the computational power of quantum computers for practically-relevant problems, especially in machine learning. In the past several years, numerous efficient algorithms for various tasks have been found, while an analysis of lower bounds is still missing.
Using communication complexity, in this work we propose the first method to study lower bounds for these tasks. We mainly focus on lower bounds for solving linear regressions, supervised clustering, principal component analysis, recommendation systems, and Hamiltonian simulations. For those problems,  we prove a quadratic lower bound in terms of the Frobenius norm of the underlying matrix. As quantum algorithms are linear in the Frobenius norm for those problems, our results mean that the quantum-classical separation is at least quadratic.
As a generalisation, we extend our method to study lower bounds analysis of quantum query algorithms for matrix-related problems using quantum communication complexity. Some applications are given.

\end{abstract}

\section{Introduction}

Initiated by Tang in the seminal work \cite{tang2019quantum}, quantum-inspired classical algorithms now provide a new way for us to understand the power and limitations of quantum computers for some practically relevant problems, especially in machine learning \cite{chia2022sampling}. 
These algorithms are very different from standard classical algorithms that usually output vector solutions. Instead, they are more comparable to quantum algorithms and can be used as a tool to study quantum advantage.

Roughly, in a quantum-inspired classical algorithm for a matrix-relevant problem, such as a linear regression problem $\arg\min_{\x} \|A\x-\b\|$, we are given access to a matrix $A$ and vector $\b$ via some data structures (similar to QRAM access), and the goal is to output a data structure to the solution $\x$. The requirements of this data structure are that we must be able to query entries of the solution and sample from the distribution defined by the solution under the $\ell_2$-norm. This data structure can be thus be viewed as a classical analog of the quantum state of the solution. 

The importance of quantum-inspired classical algorithms lies in the finding that quantum computers only achieve polynomial speedups rather than exponential speedups as we originally expected for some matrix problems when the input matrices have low rank.
A typical example is the linear regression problem. Assuming the quantum random access memory (QRAM) data structure, the best quantum algorithm known so far has complexity linear in $\|A\|_{\F}$, which is believed to be optimal \cite{chakraborty_et_al:LIPIcs:2019:10609}. Here $\|A\|_{\F}$ is the Frobenius norm of $A$, which is defined as the square root of the sum of the absolute squares of all entries. Using a similar data structure to QRAM, some efficient quantum-inspired classical algorithms were also proposed \cite{bakshi2023improved, shao2022faster, gilyen2022improved, chia2022sampling}.
The best algorithms known so far have complexity quartic in $\|A\|_{\F}$ \cite{bakshi2023improved, shao2022faster}. This implies that the quantum speedup is at most quartic in terms of the Frobenius norm. We remark that the efficiency of quantum/quantum-inspired algorithms is also influenced by other parameters, however, the Frobenius norm is usually the dominating one.

To better understand quantum-classical separations and also to see how far these quantum-inspired classical algorithms are far from optimal, we provide the first approach for proving lower bounds of quantum-inspired classical algorithms. 

\subsection{Main results}

Our main focus of this work is on the lower bounds analysis of quantum-inspired classical algorithms for five problems: linear regressions, supervised clustering, principal component analysis, recommendation systems, and Hamiltonian simulation. Algorithms for these problems are well-studied both in the quantum case (e.g., \cite{lloyd2013quantum,chakraborty_et_al:LIPIcs:2019:10609,kerenidis2017quantum,gilyen2019quantum,lloyd2014quantum-pca,harrow2009quantum}, to name a few) and in the quantum-inspired classical case (e.g., \cite{chia2022sampling,bakshi2023improved,tang2021quantum,gilyen2018quantum,gharibian2022dequantizing}, to name a few). 
There are lots of results on the lower bounds of quantum algorithms \cite{gilyen2019quantum,harrow2009quantum} (also see relevant references in \cite{montanaro-shao2023}).
However, lower bounds for quantum-inspired classical algorithms are missing.

Our main technique for lower bounds analysis is communication complexity. We will establish a connection between quantum-inspired classical algorithms and communication complexity, then use some well-known results in communication complexity (such as lower bounds for the Set-Disjointness) to obtain some lower bounds.
As a by-product, we also extend our method to study lower bounds analysis of quantum query algorithms for functions of matrices (see Section \ref{section: quantum connection}).

We now give more details. In quantum-inspired classical algorithms, we assume that we are given access to inputs via certain data structures called ``SQ''. Here S and Q stand for Sampling and Query, respectively, and the rigorous definition of ``SQ'' is given in Section \ref{section:Preliminaries}. For example, for a vector $\v=(v_1,\ldots,v_n)$, $S(\v)$ allows us to sample from the distribution defined by $\v$ under the $\ell_2$ norm (i.e., the probability of seeing $i$ is $|v_i|^2/\|\v\|^2$), and $Q(\v)$ allows us to query an entry and/or the norm of $\v$.

Our main results are briefly summarised in Table \ref{table2}. From the table, we can claim that for these tasks the quantum speedup is at least quadratic. Apart from the supervised clustering problem, all other bounds are not tight so far with respect to $\|A\|_\F$. So it is interesting to know for these problems (and/or other problems not included), whether the lower bounds can be improved or more efficient quantum-inspired classical algorithms can be found.

\setlength{\arrayrulewidth}{0.2mm}
{\renewcommand
\arraystretch{1.5}
\begin{table}
\centering
\begin{tabular}{|c|c|c|c|c|} 
 \hline
\multirow{2}{*}{Problems} & Upper bounds & Upper bounds  & Lower bounds \\
& (Q) & (QIC) & (QIC, this paper) \\ \hline
Linear Regression & $\widetilde O(\|A\|_\F)$ \cite{chakraborty_et_al:LIPIcs:2019:10609} & $\widetilde O(\|A\|_\F^4)$ \cite{bakshi2023improved} & $\widetilde \Omega(\|A\|_\F^2)$ \\
Supervised clustering & $\widetilde O(\|A\|_\F^2\|\b\|^2/\varepsilon )$ \cite{lloyd2013quantum} & $\widetilde O(\|A\|_\F^4\|\b\|^4/\varepsilon^2 )$ \cite{chia2022sampling} & $\widetilde \Omega(\|A\|_\F^4\|\b\|^4/\varepsilon)$ \\
Principal component analysis & $\widetilde O(\|A\|_\F)$ \cite{chakraborty_et_al:LIPIcs:2019:10609} & $\widetilde O(\|A\|_\F^6)$ \cite{chia2022sampling} & $\widetilde \Omega(\|A\|_\F^2)$  \\
Recommendation systems & $\widetilde O(\|A\|_\F)$  \cite{kerenidis2017quantum} & $\widetilde O(\|A\|_\F^4)$ \cite{bakshi2023improved} & $\widetilde \Omega(\|A\|_\F^2)$  \\
Hamiltonian simulation & $\widetilde O(\|A\|_\F)$ \cite{gilyen2019quantum} & $\widetilde O(\|A\|_\F^4)$ \cite{bakshi2023improved} & $\widetilde \Omega(\|A\|_\F^2)$ \\ \hline
\end{tabular}
\caption{Upper and lower bounds for quantum (Q) and quantum-inspired classical (QIC) algorithms for some other problems. We mainly show the dependence on $\|A\|_\F$ here. The explicit statements of these problems are given in Section \ref{section:Lower bounds}.
}
\label{table2}
\end{table}
}

\subsection{Summary of main ideas}

The form of access to input data in quantum-inspired classical algorithms is very similar to that in the setting of \emph{query complexity}. However, a quantum-inspired classical algorithm has more power than a conventional query algorithm, since, for example, it can sample from the input data in a way that query algorithms can not. Thus, showing lower bounds query algorithms does not directly imply lower bounds for quantum-inspired classical algorithms. A standard model of computation that is studied, and richer than the query model, is that of \emph{communication complexity}. As it turns out, quantum-inspired classical algorithms can indeed be efficiently simulated by communication protocols. Thus, showing communication complexity lower bounds for a task implies lower bounds for quantum-inspired classical algorithms solving the same task.
This approach is inspired by \cite{montanaro2022quantum}. 

Below we sketch the main ideas behind our lower bounds for the linear regression problem.
Consider two players Alice and Bob, who are individually computationally unbounded. Alice has as input a matrix $A^{(1)}$ and a vector $\b^{(1)}$. Bob has another matrix $A^{(2)}$ and another vector $\b^{(2)}$. They can see their own inputs, but neither of the players can see the other player's input. 
Their goal is to solve the linear regression $\x_*:=\arg\min_{\x}\|A\x-\b\|$, where $A = \begin{pmatrix}
    A^{(1)} \\ A^{(2)}
\end{pmatrix}$ and $\b=\begin{pmatrix}
    \b^{(1)} \\ \b^{(2)}
\end{pmatrix}$. To solve this problem, they are allowed to communicate with each other via a pre-determined communication protocol. The communication complexity is the total number of bits used in the communication on a worst-case input.
Without loss of generality, we assume the existence of a coordinator such that the communication is 2-way between each player and the coordinator. With a coordinator, it is easy to see how to generalise this idea to the multi-player model. Also, we assume that each of the entries of $A,\b$ are specified by {\color{black}$O(\log q)$ bits}.

In the model of quantum-inspired classical algorithms, we are given $SQ(A)$ and $SQ(\b)$, and the goal is to output $SQ(\x_*)$ up to certain relative error. A quantum-inspired classical algorithm of complexity $T$ is a classical algorithm that uses $T$ applications of $SQ(A), SQ(\b)$ and an arbitrary number of other arithmetic operations independent of ``SQ''. We show below that a quantum-inspired classical algorithm of complexity $T$ for the above linear regression implies a communication protocol of complexity {\color{black}$O(T\log(qmn))$} for the same problem, where $m$ and $n$ are the total number of rows and columns of $A$, respectively.  To see this, it suffices to demonstrate a protocol for the coordinator to implement $SQ(A)$ and $SQ(\b)$.

Let us first consider how the coordinator implements $SQ(\b)$. Recall that $SQ(\b)$ means sampling from the distribution defined by $\Pr(i) = |b_i|^2/\|\b\|^2$. The protocol is as follows: 
\begin{itemize}
    \item Firstly, Alice sends norm $\|\b^{(1)}\|$ and Bob sends norm $\|\b^{(2)}\|$ to the coordinator. At this point the coordinator knows $\|\b\|^2=\|\b^{(1)}\|^2+\|\b^{(2)}\|^2$.
    \item Secondly, the coordinator samples an index from the distribution defined by $\Pr(i) = \frac{\|\b^{(i)}\|^2}{\|\b\|^2}$.
    \item If the output is 1, then the coordinator asks Alice to sample an index from the distribution defined by $\b^{(1)}$ under the $\ell_2$ norm.
    Otherwise, the coordinator asks Bob to sample an index from the distribution defined by {\color{black}$\b^{(2)}$} under the $\ell_2$ norm.
\end{itemize}   
Now, it is not hard to check that the result is a desired sample. Next, let us consider how the coordinator implements $Q(\b)$, i.e., query the $i$th entry of $\b$, say. This task is simpler.
Firstly, Alice sends the dimension, say $m_1$, of $\b^{(1)}$ to the coordinator. Secondly, if the index $i\leq m_1$, then the coordinator asks Alice to send back the $i$-th entry of $\b^{(1)}$. Otherwise, the coordinator asks Bob to send back the $(i-m_1)$-th entry of $\b^{(2)}$. The protocol for the coordinator to use $SQ(A)$ is similar.

As a result, if we know the lower bounds of communication complexity of solving (say) linear regressions, we then can obtain lower bounds of quantum-inspired classical algorithms for the same task. With this idea, we then only need to focus on the lower bounds analysis of the communication complexity of linear regressions.
For this, we reduce some well-studied problems like the Set-Disjointness problem  to some matrix inversion problems derived from the original linear regression problem.

\subsection{Outline of the paper}

The rest of the paper is organised as follows: Section \ref{section:Preliminaries} contains some preliminary results on quantum-inspired classical algorithms and communication complexity. Section \ref{section:the connection} builds the formal connection between quantum-inspired classical algorithms and communication complexity. In Section \ref{section:Lower bounds} we show lower bounds for quantum-inspired classical algorithms. Finally, as a generalisation of the above idea, Section \ref{section: quantum connection} studies the connection between quantum query complexity and quantum communication complexity of functions of matrices. Some applications are introduced. Appendix \ref{section:another connection} describes another connection between quantum-inspired classical algorithms and communication complexity, which might be useful in future applications.

\section{Preliminaries}
\label{section:Preliminaries}

\subsection{The model for quantum-inspired classical algorithms}

In this section, we recall definitions of the model used in quantum-inspired classical algorithms. For more details, we refer the reader to \cite{chia2022sampling}.

\begin{defn}[Query access]
For a vector $\v=(v_1,\ldots,v_n) \in \mathbb{C}^n$, we have $Q(\v)$, \emph{query access to $\v$}, if for all $i \in [n]$, we can query for $v_i$. Likewise, for a matrix $A=(A_{ij}) \in \mathbb{C}^{m\times n}$, we have $Q(A)$ if for all $(i,j) \in [m] \times [n]$, we can query for $A_{ij}$. 
\end{defn}

\begin{defn}[Sampling and query access to a vector] For a vector $\v=(v_1,\ldots,v_n) \in \mathbb{C}^n$, we have $SQ(\v)$, \emph{sampling and query access to $\v$}, if we can

\begin{itemize}
    \item query for entries of $\v$ as in $Q(\v)$;
    \item obtain independent samples of indices $i \in [n]$, each distributed as $\Pr(i)=|v_i|^2/\|\v\|^2$;
    \item query for $\|\v\|$.
\end{itemize}
\end{defn}

\begin{defn}[Sampling and query access to a matrix] 
\label{defn:Sampling and query access to a matrix}
For a matrix $A \in \mathbb{C}^{m\times n}$, we have
$SQ(A)$ if we have $SQ(A_{i*})$ for all $i \in [m]$ and $SQ(\a)$ for $\a=(\|A_{1*}\|,\ldots,\|A_{m*}\|)$. Here $A_{i*}$ refers to the $i$-th row of $A$.
\end{defn}

By a quantum-inspired classical algorithm of (query) complexity $T$ for solving a matrix-related problem, e.g., solving a linear regression $\min\|A\x-\b\|$, we mean we ``solve'' the problem by obtaining the ``SQ'' data structure of the solution, with $T$ applications of $SQ(A), SQ(\b)$ and an arbitrary number of other arithmetic operations that are independent of ``SQ''.

\subsection{Communication complexity}

Communication complexity was first introduced by Andrew Yao more than four decades ago in the study of distributed computation \cite{yao1979some}. It has wide applications in many areas, especially in showing lower bounds, which has extensively been studied both in the classical and quantum fields \cite{de2002quantum,rao2020communication}. There are many models to study communication complexity in. In this work, we will focus on the coordinator model \cite{phillips2016lower}.

In the coordinator model, there are $k\geq 2$ players $\P_1,\ldots,\P_k$ and a coordinator $\C$. Each player holds some private information, and their goal is to solve some problem using as little communication as possible. The communication occurs between a player and the coordinator via a 2-way private channel. The computation is in terms of rounds: at the beginning of each round, the coordinator sends a message to one of the $k$ players, and then that player sends a message back to the coordinator. In the end, the coordinator returns an answer. Unless mentioned otherwise, the protocol must output a correct answer with probability at least $2/3$. The communication complexity is defined to be the total number of bits sent through the channels, on a worst-case input and worst-case outcomes of the internal randomness of the protocol.

A fundamental example is the Set-Disjointness problem. When $k=2$, player $\P_1$ receives an $n$-bit string $\x=(x_1,\ldots,x_n)\in\{0,1\}^n$ and player $\P_2$ receives another $n$-bit string $\y=(y_1,\ldots,y_n)\in\{0,1\}^n$, their goal is to determine if there is an $i$ such that $x_i=y_i=1$. It is well-known that the naive protocol ($\P_1$ has to send all $\x$ to the coordinator $\C$ who then sends $\x$ to $\P_2$, and finally $\P_2$ sends the right answer to $\C$) is asymptotically optimal \cite{kalyanasundaram1992probabilistic, razborov1990distributional}. Namely, the communication complexity of the Set-Disjointness problem is $\Theta(n)$. 

Another fundamental problem is the Gap-Hamming problem. In this problem, Alice has a $n$-bit string $\x$, Bob has a $n$-bit string $\y$, and they wish to compute
\be
\label{defn:gap-hamming}
G(\x,\y)
=\begin{cases}
    1 & \text{if } \Delta(\x,\y)\geq n/2+\varepsilon n, \\
    0 & \text{if } \Delta(\x,\y)\leq n/2-\varepsilon n, \\
    * & \text{otherwise}.
\end{cases}
\ee
In the above, $\Delta(\x,\y)$ computes the hamming distance of $\x,\y$.
As shown in \cite{chakrabarti2011optimal} the communication complexity for this problem is {\color{black}$\Theta(\min(n, 1/\varepsilon^2))$}.

For those problems, there are also some multiplayer versions. We state them as follows.


\begin{defn}[$k$-player Set-Disjointness problem]
\label{defn:k-DISJ}
For $i\in[k]$, player $\P_i$ receives a bit string $T_i=(T_{i1}, \ldots, T_{in})\in\{0,1\}^n$. Their goal is to determine if there is a $j\in\{2,\ldots,k\}$ such that $T_{1\ell}=T_{j\ell}=1$ for some $\ell\in[n]$.
\end{defn}

\begin{prop}[Theorem 3.3 of \cite{phillips2016lower}, Theorem 1 of \cite{woodruff2017distributed}]
\label{prop:k-DISJ}
When $n \geq 3200k$, for any classical protocol that succeeds with probability $1-1/k^3$ for solving $k$-player Set-Disjointness problem, the randomised communication complexity is $\Theta(k n)$.
\end{prop}

The above result holds even if all $T_i$ have Hamming weight $\Theta(n)$. This indeed follows from the proof of the above proposition in the cited references. Here we only highlight the main idea. More details can be found in \cite{phillips2016lower,woodruff2017distributed}. We use the symmetrization argument to reduce the $k$-player Set-Disjointness problem to a standard 2-player Set-Disjointness problem.
In the symmetrization argument, Alice plays the role of one player and Bob plays the role of all other players as well as the coordinator.
The 2-player Set-Disjointness problem is well studied and is known hard even if the Hamming weight of the inputs has order $n$. By the symmetrization argument, the complexity of the original $k$-player Set-Disjointness problem is $\Omega(k n)$ when $|T_j|=\Theta(n)$ for all $j$. Indeed, in the symmetrization argument, all inputs have Hamming weight $\Theta(n)$, see the reduction at 
\cite[page 8]{phillips2016lower}. Additionally, similar to the $2$-player Set-Disjointness problem, the $k$-player Set-Disjointness problem is also hard when there is at most one $j$ such that $T_{1\ell} = T_{j \ell}=1$.

\begin{defn}[$k$-Gap-Hamming problem]
\label{defn:k-gap}
For $i\in \{1,2,\ldots,k+1\}$, player $\P_i$ receives a bit string $T_i\in\{\pm 1\}^n$ with the promise that $T := \sum_{i=1}^k T_i\in \{\pm 1\}^n$ and $\sum_{j=1}^k T_{j}\cdot T_{k+1} \in [-c_2\sqrt{n}, c_2\sqrt{n}]$.
Their goal is to determine if $\sum_{i=1}^k T_i\cdot T_{k+1} \geq c_1\sqrt{n}$ or $\leq -c_1\sqrt{n}$. If $\sum_{i=1}^k T_i\cdot T_{k+1}$ is in between then any output is allowed.
\end{defn}

\begin{prop}[Theorem 3.3 of \cite{li2023ell_p}]
\label{prop:k-gap}
Any protocol that solves the $k$-Gap-Hamming problem with large constant probability requires $\Omega(kn)$ bits of communication.
\end{prop}

Let $P=(p_1,\ldots,p_n)$ be a distribution. When we say that we \emph{approximately sample} from $P$, we mean that we output a sample from any distribution $\widetilde{P}=(\tilde{p}_1,\ldots,\tilde{p}_n)$ such that $\|\widetilde{P}-P\|_1:=\sum_i |p_i - \tilde{p}_i| \leq \varepsilon$ for some constant inaccuracy $\varepsilon$.

\begin{defn}[Distributed sampling problem]
\label{defn:DSP}
Assume that Alice has as input a function $f:\{0,1\}^n \rightarrow \{\pm 1\}$ and
Bob has as input a function $g:\{0,1\}^n \rightarrow \{\pm 1\}$.
Their goal is to approximately sample from the distribution defined by
\[
\Pr(y) := \left(\frac{1}{2^n}  \sum_{x\in\{0,1\}^n} f(x) g(x) (-1)^{x\cdot y} \right)^2.
\]
\end{defn}

\begin{prop}[Theorem 1 of \cite{montanaro2019quantum}] 
\label{prop:DSP}
There exists a universal constant $\varepsilon$ such that, for sufficiently large $n$, any 2-way classical communication protocol for Distributed Fourier Sampling with shared randomness and inaccuracy $\varepsilon$ must communicate at least $\Omega(2^n)$ bits.
\end{prop}

\subsection{Some background on quantum computing}
\label{subsection:Some background on quantum computing}

In this paper, we often use the Dirac notation to represent unit vectors. 
For a unit column vector $\v$, we sometimes write it as $\ket{\v}$. In the quantum case, it will be called a quantum state.
Its complex conjugate is denoted as $\bra{\v}$.
The standard basis of $\mathbb{C}^n$ will be denoted as $\ket{1}, \ket{2}, \cdots, \ket{n}$.
Let $\ket{\a}=\sum_i a_i \ket{i}$ be a quantum state. The coefficients $a_i$ must satisfy $\sum_{i = 1}^n |a_i|^2=1$. When we measure it in the standard basis, then we will observe $i$ with probability $|a_i|^2$. So measurement naturally corresponds to the operation of sampling. 
In this paper, the quantum notation $\ket{0}^n$ will also be used. 
This is the same as $\ket{1}$ in $\mathbb{C}^{2^n}$.
For more about quantum computation, we refer the reader to the book \cite{nielsen2010quantum}.

\section{The connection}
\label{section:the connection}

We below discuss the connection between quantum-inspired classical algorithms and communication complexity. Our main result is the following theorem, which shows that there is an efficient simulation of a quantum-inspired classical algorithm using a classical communication protocol.

\begin{thm}
\label{thm:connection1}

In the multi-player coordinator model, for each $i\in\{1,\ldots,k\}$, assume that player $\P_i$ holds a matrix $A^{(i)}\in \mathbb{R}^{\ell_i\times n}$ and a vector $\b^{(i)}\in \mathbb{R}^{m_i}$ with $m:=\sum_i \ell_i = \sum_i m_i$.
Assume that all entries are specified by $O(\log q)$ bits.
Let
\be
\label{thm:eq}
A=\begin{pmatrix}
A^{(1)} \\
\vdots \\
A^{(k)} \\
\end{pmatrix}_{m\times n}, 
\quad
\b=\begin{pmatrix}
\b^{(1)} \\
\vdots \\
\b^{(k)}
\end{pmatrix}_{m\times 1}.
\ee
Then we have the following:
\begin{itemize}
\item The coordinator $\C$ can use $SQ(A)$ $O(T)$ times, using $O((T+k) \log (qmn))$ bits of communication.
\item The coordinator $\C$ can use $SQ(\b)$ $O(T)$ times, using $O((T+k) \log(qm))$ bits of communication.
\end{itemize}

\end{thm}

\begin{proof}
We will prove the second claim first.
For the coordinator $\C$ to use $SQ(\b)$, i.e., sample from and query for $\b$, the player $\P_i$ first sends $\|\b^{(i)}\|$ to $\C$. This requires $O(k\log (qm))$ bits of communication. Here $O(\log(qm))$ is the number of bits used to specify $\|\b^{(i)}\|$.
Notice that this only needs to be done once for the whole protocol, rather than once for each access.
At this point, the coordinator knows $\|\b\|$, and also can sample from the distribution
\[
\D_{\b} (i) = \frac{\|\b^{(i)}\|^2}{\|\b\|^2},
\quad i\in\{1,2,\ldots,k\}.
\]
To sample from $\b$, the coordinator first uses the distribution $\D_{\b}$ to obtain a random index $i$, then communicates with $\P_i$ and asks $\P_i$ to sample from the distribution defined by $\b^{(i)}$ under the $\ell_2$-norm. For an arbitrary index $j \in [m]$, suppose the $j$'th entry of $\b$ is held by player $i$, and say {\color{black}$b^{(i)}_k = b_j$}. Then it is easy to see that
\bes
\Pr(j) = \frac{\|\b^{(i)}\|^2}{\|\b\|^2} \frac{|b^{(i)}_k|^2}{\|\b^{(i)}\|^2}
=\frac{|b_j|^2}{\|\b\|^2},
\ees
which is the same as the distribution defined in $SQ(\b)$. 

To query the $j$-th entry of $\b$, the coordinator can communicate with the player who holds index $j$. To this end, player $\P_i$ only needs to send $m_i$ to $\C$, which costs $O(\log m)$ bits of communication. 
So in total $O(k\log m)$ bits of communication. In conclusion, for the coordinator $\C$ to use $SQ(\b)$ $O(T)$ times, they use $O(T\log (qm)+k\log (qm) )$ bits of communication. Here $O(T \log(qm))$ counts the total number of bits of receiving indices or entries from all players, and $O(k \log(qm))$ counts the total number of bits of receiving the norm $\|\b^{(i)}\|$ and $m_i$ from player $\P_i$ for all $i$.

Similarly, for the coordinator to use $SQ(A)$, the player $\P_i$ first sends $\|A^{(i)}\|_{\F}$ to the coordinator $\C$. This costs $O(k\log (qmn))$ bits of communication, where $O(\log(qmn) )$ is the number of bits used to specify $\|A^{(i)}\|_\F$. The coordinator now has access to the distribution
\bes
\D_A(i) = \frac{\|A^{(i)}\|_{\F}^2}{\|A\|_{\F}^2}, \quad i\in\{1,2,\ldots,k\}.
\ees
 As in Definition \ref{defn:Sampling and query access to a matrix}, let $\a=(\|A_{1*}\|,\cdots,\|A_{m*}\|)$, where $A_{i*}$ is the $i$'th row of $A$. The coordinator can use $SQ(\a)$ in a similar way to $SQ(\b)$. Now each entry is specified by $O(\log (qn))$ bits. To use $SQ(A_{i*})$, the coordinator only needs to communicate with the player who has the row index $i$. In summary, for the coordinator to use $SQ(A)$ $O(T)$ times, they use $O((T+k) \log(qmn))$ bits of communication.
\end{proof}

Recall that a quantum-inspired classical algorithm of complexity $O(T)$ means an algorithm that uses $O(T)$ applications of the data structure ``SQ'' and an arbitrary number of other arithmetic operations. As a result of Theorem \ref{thm:connection1}, if there is a quantum-inspired classical algorithm of complexity $O(T)$, then there is a protocol for the same problem with communication complexity $O((T+k) \log(qmn))$. So we can use this to prove lower bounds of $T$.
To obtain a non-trivial lower bound from this result, it is useful to consider the regime where $k$ and $q$ are small. This is usually easy to enforce in the multi-player model {\color{black} when $k\ll T$. Indeed, for practical applications below, we can design appropriate models such that $q=O(m+n)$ and $k, \log(q) \ll T$. Thus the communication complexity is mainly dominated by $T$.}

\begin{rem} \label{remark}
{\rm    
In Theorem \ref{thm:connection1}, it is possible that $A^{(i)}$ or $\b^{(i)}$ is empty or public for some $i$.
The above argument still works.
A very special case is that the whole matrix $A$ or the whole vector $\b$ belongs to a particular player, say player $i$. In this case, for the coordinator to use $SQ(A)$ or $SQ(\b)$, it suffices to communicate with player $i$. 
}
\end{rem}

Note that there are already many efficient quantum-inspired classical algorithms for a wide range of problems, e.g., see \cite{bakshi2023improved,chia2022sampling}, so we can use these algorithms to construct efficient communication protocols for the same problems based on Theorem \ref{thm:connection1}. One big advantage of this is that when studying communication complexity, the assumptions on quantum-inspired classical algorithms can be greatly relaxed. More precisely, one costly step of quantum-inspired classical algorithms is the preprocessing step of building the data structure ``SQ'' for the input matrices or vectors. However, in communication complexity, local costs are not considered since the individual players are considered to be all-powerful. So each player can build the data structure ``SQ'' to their own inputs independently, and this cost will not be counted when analyzing the communication complexity (note that the only power that we needed the players to have in the proof of Theorem~\ref{thm:connection1} was SQ access to their own inputs). 
Namely, the preprocessing step can be resolved easily in the model of communication complexity.
This also means that quantum-inspired classical algorithms can play important roles in distributed computation.
To support this point, we introduce another connection in Appendix \ref{section:another connection}, which we expect to be more useful for applications. However, we will not consider specific applications of this in the current paper.

\section{Lower bounds for quantum-inspired classical algorithms}
\label{section:Lower bounds}

In this section, we use Theorem \ref{thm:connection1} to prove lower bounds for quantum-inspired classical algorithms for five problems: linear regression, supervised clustering, principal component analysis, recommendation systems, and Hamiltonian simulation.  Below, the statements of these problems come from \cite{chia2022sampling}.

\begin{prob*}[Linear regression]
    Let $A$ be a matrix and $\b$ be a vector, given $SQ(A)$ and $SQ(\b)$, output $SQ(\tilde{\x}_*)$ such that $\|\tilde{\x}_*-\x_*\|\leq \varepsilon \|\x_*\|$, where $\x_*=A^+\b$.
\end{prob*}

From \cite{bakshi2023improved}, the complexity of the best quantum-inspired classical algorithm known so far for linear regression is $\widetilde{O}(\kappa_\F^4 \kappa^{10}/\varepsilon^2 \gamma^2)$, where $\kappa_\F, \kappa, \gamma$ are defined as follows:
Let $\sigma_{\min}$ be the minimal nonzero singular value of $A$, then
\be
\label{some notation}
\kappa_{\F}:=\|A\|_\F/\sigma_{\min}, \quad 
\kappa := \|A\|/\sigma_{\min}, \quad 
\gamma := \|A\x_*\|/\|\b\|.
\ee
These are the scaled condition number, condition number and the overlap of $\b$ is the column space of $A$, respectively.
In the special case when $A$ is row-sparse, the complexity is only $O(s\kappa_\F^2\log(1/\varepsilon))$ when $\gamma=\Theta(1)$ \cite{shao2022faster}, where $s$ is the row sparsity (i.e., the maximal number of nonzero entries of each row). Thus there are five parameters we have to consider in the lower bounds analysis, i.e., $\kappa_\F, \kappa, \varepsilon, \gamma, s$. 
In the lower bounds analysis below, we will mainly focus on the analysis with respect to $\kappa_\F$, which relates the rank of $A$ and is usually the dominating term in the complexity. Indeed, if the rank of $A$ is $r$, then it is obvious that $\kappa_\F \geq \sqrt{r}$ because {\color{black}$\kappa_\F^2=\|A\|_\F^2/\sigma_{\min}^2$ and $\|A\|_\F^2$} equals the square sum of all singular values. So $\kappa_\F$ can be very large even if the matrix is well-conditioned.

With Theorem \ref{thm:connection1}, it suffices to give lower bounds on the communication complexity of linear regressions where $A,\b$ are defined in the form of (\ref{thm:eq}) in the coordinator model.
As our main focus will be on the dependence on $\kappa_\F$, in our reductions we try to ensure that all the other parameters $\kappa, \varepsilon, \gamma, s$ are $O(1)$, to obtain non-trivial lower bounds in terms of $\kappa_\F$. 
We use the notation $\widetilde{\Omega}(\cdot)$ to hide all factors that are polylogarithmic in the input size. Specifically, we will ignore the logarithmic factors in Theorem~\ref{thm:connection1} for the sake of readability. Below, we focus on the sampling task in SQ in two cases: row-sparse and dense cases.

\begin{prop}[Sampling in the row-sparse case]
Assume that $A$ is row sparse and $\varepsilon\in(0,1)$ is a constant. Then $\varepsilon$-approximately sampling from $A^+\b$ requires making $\widetilde{\Omega}((\kappa^2+\kappa_\F)/\gamma)$ calls to $SQ(A), SQ(\b)$.
\end{prop}

\begin{proof} We give a reduction from the $k$-player Set-Disjointness problem, defined in Definition~\ref{defn:k-DISJ}. Suppose the $k$ players $\P_1, \dots, \P_k$ are given inputs to the Set-Disjointness problem (using the notation given in Definition \ref{defn:k-DISJ}).
Player $\P_j$ constructs the following vector:
\bes
\t_j = \sum_{\ell=1}^n T_{j \ell} \ket{\ell}, \quad
\ket{\t_j} = \t_j/\alpha_j,
\ees
where $\alpha_j=\|\t_j\|= \Theta(\sqrt{n})$ for all $j$.
Here we used the ket notation $\{\ket{1},\ldots,\ket{n}\}$ to represent the standard basis of $\mathbb{C}^n$, see Subsection \ref{subsection:Some background on quantum computing}. So $\ket{\t_j}$ is a unit vector.

Let $\beta_A,\beta_b$ be parameters that will be fixed later. 
Let 
\beas
A = \beta_A \ket{1}_k \otimes \ket{1}_n \bra{1}_k + \sum_{j=2}^k \ket{j}_k  \otimes \ket{\t_j} \bra{j}_k , \quad
\b =   \beta_b \ket{1}_k \otimes \ket{1}_n + n\sum_{j=2}^k \ket{j}_k \otimes \ket{\t_1},
\eeas
where the subindex $k$ or $n$ refers to the dimension of base vectors.
In matrix form,
\[
A = \begin{pmatrix}
\beta_A \ket{1}_n & \\
& \ket{\t_2} & \\
& & \ddots \\
& & & \ket{\t_k} \\
\end{pmatrix}_{kn \times k}, 
\quad
\b = 
\begin{pmatrix}
\beta_b \ket{1}_n  \\
n \ket{\t_1}  \\
\vdots \\
n \ket{\t_1} \\
\end{pmatrix}_{kn \times 1}.
\]
In this construction, the first-row block of $A$ and $\b$ are public for all players. The second-row block of $A$ belongs to player 2, and the last-row block of $A$ belongs to player $k$. The whole vector of $\b$ except the first-row block belongs to player 1. This satisfies the conditions made in Theorem \ref{thm:connection1}, also see Remark \ref{remark}.

For the above matrix, the row sparsity is 1,
and it is easy to check that
\[
A^\T A = \beta_A^2\ket{1}_k \bra{1}_k + \sum_{j=2}^k \ket{j}_k \bra{j}_k,
\quad 
A^\T\b = \beta_A\beta_b \ket{1}_k + n \sum_{j=2}^k \langle \t_1|\t_j\rangle \ket{j}_k.
\]
The optimal solution of the linear regression $\arg\min\|A\x-\b\|$ is
\[
\x_* = (A^\T A)^{-1} A^\T\b = \frac{\beta_b}{\beta_A}\ket{1}_k + n\sum_{j=2}^k \langle \t_1|\t_j\rangle \ket{j}_k.
\]
If there is no intersection for the $k$-player Set-Disjointness problem, then $\x_* = \frac{\beta_b}{\beta_A}\ket{1}_k$. Measuring this state only returns index 1 with probability 1. If there is an intersection with one common index, then there is a $j\in\{2,\ldots,k\}$ such that
\[
\x_* = \frac{\beta_b}{\beta_A}\ket{1}_k + \frac{n}{\alpha_1 \alpha_j} \ket{j}_k.
\]
Note that $\alpha_1 \alpha_j = \Theta(n)$, so if we choose $\beta_A=\Theta(\beta_b)$, then with a constant probability we will see index $j$ by sampling from the distribution defined by $\x_*$. In conclusion, if we can sample from the solution, we then can solve the $k$-player Set-Disjointness problem using a constant number of such samples. 

In the remaining part of the proof we analyse the parameters of this problem and conclude the lower bound using Theorem~\ref{thm:connection1} and Proposition~\ref{prop:k-DISJ}.

The overlap of $\b$ in the column space of $A$ is
\[
\gamma^2 = \frac{\|A\x_*\|^2}{\|\b\|^2} = 
\frac{\beta_b^2 + n^2 \sum_{j=2}^k \langle \t_1|\t_j\rangle^2 }{\beta_b^2 + (k-1)n^2} =
\frac{\beta_b^2 + c }{\beta_b^2 + (k-1)n^2}
\]
for some constant $c=\Theta(1)$.
We also have
\[
\|A\|_\F^2 = \beta_A^2 + k-1, \quad 
\sigma_{\min} = \min\{ \beta_A, 1 \}.
\]
So 
\[
\kappa^2 = \frac{\max(\beta_A^2,1)}{\min(\beta_A^2,1)}, \quad
\kappa_\F^2 = \frac{\|A\|_\F^2}{\sigma_{\min}^2} = \frac{\beta_A^2 + k-1}{\min(\beta_A^2,1)}, \quad \gamma^2 = \frac{\beta_b^2+c}{\beta_b^2 + (k-1)n^2}.
\]

If we choose $\beta_A^2 = \beta_b^2 = k$, then
$\kappa^2, \kappa^2_\F = \Theta(k), \gamma^2 = \Theta(1/n^2)$. Now we cannot determine the dependence on $\kappa_\F$. But note that $\kappa_\F \geq \kappa$, so by Theorem \ref{thm:connection1} a lower bound we can claim is $\widetilde{\Omega}(\kappa^2/\gamma)$ because the communication complexity of the $k$-player Set-Disjointness problem is $\Theta(kn)$ by Proposition \ref{prop:k-DISJ}.
If we choose $\beta_A^2=\beta_b^2 = 1$, then $\kappa = 1, \kappa_\F^2 = \Theta(k), \gamma^2 = \Theta(1/kn^2)$, now we obtain a lower bound of $\widetilde{\Omega}(\kappa_\F/\gamma)$. 

One thing we did not analyse is the error. In the above constructions, we either have $\x_* = \ket{1}_k$ or $\x_* = \ket{1}_k + d \ket{j}_k$ with $d=\Theta(1)$ depending on the result of the Set-Disjointness problem. The error will not change the sampling result too much when the inaccuracy $\varepsilon=O(1)$ is a small constant because the two solutions are far from each other.
\end{proof}

\begin{prop}[Sampling in the dense case]
Assume that $\varepsilon\in(0,1)$ is a constant. Then $\varepsilon$-approximately sampling from $A^+\b$ requires making $\widetilde{\Omega}(\kappa_\F^2)$ calls to $SQ(A), SQ(\b)$.
\end{prop}

\begin{proof}
Here we give a reduction from the distributed sampling problem. We below use the notations stated in Definition \ref{defn:DSP}. Let $D_f =\sum_x f(x)  \ket{x}\bra{x} $ be the diagonal matrix such that the $x$-th diagonal entry is $f(x)$ and let $ \ket{g} = \frac{1}{\sqrt{2^n}} \sum_x g(x) \ket{x}$ be a unit column vector whose entries are $g(x)$, then the distribution in the distributed sampling problem is defined by the state $(D_f H^{\otimes n})^{-1} \ket{g}$, where 
$H=\frac{1}{\sqrt{2}} \begin{pmatrix}
    1 & 1 \\ 1 & -1
\end{pmatrix}$ is the Hadamard gate.
Now Alice has the whole matrix $A:=D_f H^{\otimes n}$ and Bob has the whole vector $\b:=\ket{g}$. This clearly satisfies the conditions in Theorem \ref{thm:connection1}.
It is also easy to check that
for matrix $D_f H^{\otimes n}$ we have $\kappa_\F^2 = 2^n$, $\kappa=\gamma=1$. So we obtain a lower bound of $\widetilde{\Omega}(\kappa_\F^2)$ by Proposition \ref{prop:DSP} and Theorem \ref{thm:connection1}.
\end{proof}

\begin{prob*}[Supervised clustering]
Let $\p, \q_1, \ldots,\q_n \in \mathbb{R}^d$ and $\varepsilon\in(0,1)$, denote
\[
A = \begin{pmatrix}
-\p/\|\p\| \\
\q_1/\|\q_1\|\sqrt{n} \\
\vdots \\
\q_n/\|\q_n\|\sqrt{n} 
\end{pmatrix}, \quad
\b = \begin{pmatrix}
\|\p\| \\
\|\q_1\|/\sqrt{n} \\
\vdots \\
\|\q_n\|/\sqrt{n} 
\end{pmatrix}.
\]
Given $SQ(A), SQ(\b)$, compute $\|\b^\T A\|^2 \pm \varepsilon$. 
\end{prob*}

If we denote $\q = \frac{1}{n} \sum_{i=1}^n \q_i$, then $\|\b^\T A\|^2=\|\p - \q\|^2$, which is the distance between $\p$ and the center of $\{\q_1,\ldots,\q_n\}$. This task is widely used in supervised cluster assignments. In this problem, we are given a vector $\p$ and some sets of vectors $Q_1, \ldots, Q_k$. The goal is to assign $\p$ to one of the sets. In supervised clustering, the criterion relies on the distance between $\p$ and the centers. We will assign $\p$ to $Q_i$ if the distance between $\p$ and the center of $Q_i$ is the smallest one.

For supervised clustering, by \cite[Corollary 6.10]{chia2022sampling}, there is a quantum-inspired classical algorithm with complexity $\widetilde{O}(\|A\|_\F^4\|\b\|^4\varepsilon^{-2})$. In the quantum case, by \cite{lloyd2013quantum}, there is a quantum algorithm of complexity $\widetilde{O}(\|A\|_\F^2\|\b\|^2\varepsilon^{-1})$. Below we show a near-tight lower bound of quantum-inspired classical algorithms, which implies that the quantum-classical separation for supervised clustering is quadratic.

\begin{prop}
For supervised clustering, any quantum-inspired classical algorithms must use $\widetilde{\Omega}(\|A\|_\F^4\|\b\|^4\varepsilon^{-1} + \varepsilon^{-2})$ applications of $SQ(A), SQ(\b)$.
\end{prop}

\begin{proof}
We obtain the first term in the complexity via a reduction from the $n$-Gap-Hamming problem, see Definition \ref{defn:k-gap}. Let $\x,\y_i\in \{\pm 1\}^d$ and $\alpha$ be a scaling parameter, which will be determined later. Let $\p=\alpha \x, \q_i = \alpha \y_i$. 
This construction coincides with the setting stated in Theorem~\ref{thm:connection1}.
Now we can compute that $\|A\|_\F^2=2, \|\b\|^2=2\alpha^2 d$. Moreover, $\|\p-\q\|_2^2 = 2\alpha^2 d - \frac{2\alpha^2}{n} \sum_{i=1}^n \x \cdot \y_i$. In the $n$-Gap-Hamming problem, the promise is that $\sum_{i=1}^n \x \cdot \y_i \geq \sqrt{d}$ or $\leq -\sqrt{d}$. To distinguish them it suffices to choose $\varepsilon = \alpha^2 \sqrt{d}/n$. We now set $\alpha$ such that $\alpha^2 \sqrt{d} = 1$, then $\varepsilon=1/n$ and $\|\b\|^2=2\sqrt{d}$.
By Proposition~\ref{prop:k-gap}, the lower bound of the $n$-Gap-Hamming problem is $\Omega(nd) = \Omega(\|A\|_\F^4\|\b\|^4\varepsilon^{-1})$, so we obtain the first part of the claimed lower bound using Theorem~\ref{thm:connection1}.

In particular, when $n=1$, the reduction is from the (2-player) Gap-Hamming problem, see Equation (\ref{defn:gap-hamming}). We now set $\p=\x/\sqrt{d}, \q=\y/\sqrt{d}$ as unit vectors. Then
$\|\p-\q\|_2^2 = \frac{4}{d}\Delta(\x,\y)$. If we can approximate this quantity up to additive error $\varepsilon $, we then can solve the Gap-Hamming problem. So we obtain a lower bound of $\widetilde{\Omega}(\varepsilon^{-2})$.
\end{proof}

\begin{prob*}[Principal component analysis]
\label{prob: Principal component analysis}
Let $A\in \mathbb{C}^{m\times n}$ be a matrix such that $A^\dag A$ has top $k$ eigenvalues $\{\lambda_i:i\in[k]\}$ and corresponding eigenvectors $\{\v_i:i\in[k]\}$. Given $SQ(A)$, with probability $\geq 1-\delta$, compute eigenvalues $\{\tilde{\lambda}_i:i\in[k]\}$ 
such that $\sum_{i\in[k]} |\tilde{\lambda}_i - \lambda_i|\leq \varepsilon \|A\|_\F^2$ and eigenvectors $\{SQ(\title{\v}_i):i\in[k]\}$ such that $\|\v_i - \tilde{\v}_i\| \leq \varepsilon$.
\end{prob*}

About this problem, by \cite[Corollary 6.12]{chia2022sampling} there is a quantum-inspired classical algorithm whose complexity is $\widetilde{O}( \frac{\|A\|_{\F}^6}{\lambda_k^2 \|A\|^2} \eta^{-6} \varepsilon^{-6} )$, where $\eta = \min_{i\in[k]} |\lambda_i-\lambda_{i+1}|/\|A\|^2$. In particular, when $k=1$, the complexity becomes $\widetilde{O}( \frac{\|A\|_{\F}^6\|A\|^6} {|\lambda_1-\lambda_2|^6}\varepsilon^{-6} )$. In comparison, the quantum complexity is $\widetilde{O}(\|A\|_\F/\varepsilon)$ when $k=1$ \cite{chakraborty_et_al:LIPIcs:2019:10609}. In this particular case, we can prove the following lower bound for quantum-inspired classical algorithms. It indicates the quantum-classical separation is at least quadratic.

\begin{prop}
\label{prop:PCA1}
Let $A$ be a matrix with maximal singular value $\sigma$ and corresponding right singular vector $\v$, then computing $\sigma\pm \varepsilon$ or $\varepsilon$-approximately sampling from $\v$ requires $\widetilde{\Omega}(\|A\|_\F^2)$ applications of $SQ(A)$.
\end{prop}

\begin{proof}
We show a reduction from the Set-Disjointness problem.
In the Set-Disjointness problem, Alice has $(a_1,\ldots,a_n) $ $ \in \{0,1\}^n$ and Bob has $(b_1,\ldots,b_n) \in \{0,1\}^n$. The goal is to determine if there is an $i$ such that $a_i=b_i=1$. If $i$ exists, we can assume that there is only one such $i$. For this problem, the lower bound is $\Omega(n)$. To use this result, Alice constructs a diagonal matrix $D_a = {\rm diag}(a_1,\ldots,a_n)$, and Bob constructs another diagonal matrix $D_b = {\rm diag}(b_1,\ldots,b_n)$. Let $A=\begin{pmatrix}
    D_a \\
    D_b
\end{pmatrix}$. Then we have $A^\dag A = D_a^2 + D_b^2$. If there is an intersection, say $a_i=b_i=1$, then $\sigma=\sqrt{2}$ and $\v=\ket{i}$. Otherwise $\sigma=1$. The corresponding right singular vectors are $\v=\ket{j}$ for any $j$ with $a_j=0,b_j=1$ or $a_j=1,b_j=0$. So if we can approximate $\sigma$ or sampling from $\v$, we then can solve the Set-Disjointness problem. In this construction, we have $\|A\|_\F^2 = \Theta(n), \|A\|=\Theta(1), \varepsilon=\Theta(1)$. So we obtain the claimed lower bound by Theorem~\ref{thm:connection1}.
\end{proof}

Recall that if the SVD of $A$ is $\sum_i \sigma_i \ket{u_i} \bra{v_i}$, then we define $A_{\geq \delta} := \sum_{i : \sigma_i\geq \delta}\sigma_{i} \ket{u_i} \bra{v_i}$.

\begin{prob*}[Recommendation systems]
\label{prob: Recommendation systems}
Let $A\in \mathbb{C}^{m\times n}$ be a matrix and $i\in[m]$. Let $\varepsilon, \delta >0$. Given $SQ(A)$, the goal is to $\varepsilon$-approximately sample from the $i$-th row of $A_{\geq \delta}$.
\end{prob*}

For the problem of recommendation systems, by \cite[Corollary 1.2]{bakshi2023improved}, there is a quantum-inspired classical algorithm of complexity $O(\|A\|_{\F}^4/\sigma^8\varepsilon^2)$. In comparison, there is a quantum algorithm of complexity $O(\|A\|_\F/\sigma)$ \cite{kerenidis2017quantum}. As a corollary of Proposition \ref{prop:PCA1}, we have the following lower bound.

\begin{prop}
Approximately sample from the $i$-th row of $A_{\geq \delta}$ requires making $\widetilde{\Omega}(\|A\|_\F^2)$ calls to $SQ(A)$. 
\end{prop}

\begin{proof}
We use the same construction as in the proof of Proposition \ref{prop:PCA1}. Now we choose $\delta \in (1,\sqrt{2})$, then $A_{\geq \delta} = \sigma \ket{\u} \bra{\v}$ for some $\u$ if there is an intersection, and $A_{\geq \delta} = 0$ if no intersection. So if we can sample from the $i$-th row of $A_{\geq \delta}$ approximately, we then can solve the Set-Disjointness problem.
\end{proof}

\begin{prob*}[Hamiltonian simulation]
\label{prob: Hamiltonian simulation}
Let $A\in \mathbb{C}^{n\times n}$ be a Hermitian matrix with $\|A\|\leq 1$ and $\v\in\mathbb{C}^{n\times n}$ be a unit vector. Let $\varepsilon, \delta \in (0,1]$. Given $SQ(A), SQ(\v)$, output $SQ(\u)$ with probability at least $1-\delta$ such that $\|\u-e^{iA t} \v\|\leq \varepsilon$.
\end{prob*}

From \cite[Corollary 1.6]{bakshi2023improved}, we know that there is a quantum-inspired classical algorithm for Hamiltonian simulation with complexity $\widetilde{O}(\|A\|_\F^4 t^{8}/\varepsilon^2)$. In comparison, the quantum algorithm has complexity $\widetilde{O}(\|A\|_\F t)$. Regarding this problem, we have the following lower bound:

\begin{prop}
Assume that $\varepsilon\in(0,1)$ is a constant. Then $\varepsilon$-approximately sampling from {\color{black}$e^{iAt} \v$} requires making $\widetilde{\Omega}(\|A\|_\F^2)$ calls to $SQ(A), SQ(\v)$.
\end{prop}

\begin{proof}
The result follows via a reduction from the distributed Fourier sampling problem. In this problem, Alice has a function $f:\{0,1\}^n \rightarrow \{\pm 1\}$ and Bob has a function $g:\{0,1\}^n \rightarrow \{\pm 1\}$. Let $D_f, D_g$ be the diagonal matrices defined by $f(x),g(x)$ respectively. Let $H_2$ be the Hadamard gate. Then their goal is to sample from $D_f H_2^{\otimes n} D_f D_g H_2^{\otimes n} \ket{0}^n$. By Proposition \ref{prop:DSP}, the communication complexity of this problem is $\Theta(2^n)$.

Let
\[
A = D_f \left(\frac{1}{2n}  \sum_{j=1}^n I_2^{\otimes (j-1)} \otimes (I_2-H_2) \otimes I_2^{\otimes(n-j)} \right) D_f,
\]
then it is easy to check that $D_f H_2^{\otimes n}D_f = e^{i A t}$, where $t=n\pi$. This matrix is in Alice's hand. For this matrix, it is not hard to show that $\|A\|_\F^2 = \Theta(2^n), \|A\|=1$. In the Hamiltonian simulation, we set $\v=D_g H_2^{\otimes n} \ket{0}^n$. This vector is in Bob's hand. Finally, by Proposition \ref{prop:DSP} and Theorem~\ref{thm:connection1}, we obtain a lower bound of $\widetilde{\Omega}(\|A\|_\F^2)$.
\end{proof}

\section{A connection between quantum query and communication complexities of matrix problems}
\label{section: quantum connection}

In this section, we generalise Theorem \ref{thm:connection1} to the quantum case. We focus on two types of quantum algorithms for matrix-based problems.

\begin{description}
\item[The query model:] 
For a matrix $A = (A_{ij}) \in \mathbb{C}^{m\times n}$, we assume that there is a quantum oracle that can query its entries in the form of
\be
\label{oracleA}
\mathcal{O}_A:  \ket{i,j} \ket{b} \mapsto  \ket{i,j} \ket{b\oplus A_{ij}}.
\ee
For sparse matrices, two additional oracles that reveal the information of positions of nonzero entries are given, i.e.,
\bea
&& \mathcal{O}_R : \ket{i} \ket{j} \mapsto \ket{i} \ket{r_{ij}}, 
\label{oracleR} \\
&& \mathcal{O}_C : \ket{i} \ket{j} \mapsto \ket{c_{ij}} \ket{j},
\label{oracleC}
\eea
where $r_{ij}$ is the index for the $j$-th non-zero entry of the $i$-th row of $A$, and if there are less than $s$ non-zero entries, then it is $j+n$; $c_{ij}$ are defined similarly according to columns.


In this setting, we focus on quantum query algorithms of the following form:
\be
\label{quantum alg1}
U_T \mathcal{O} U_{T-1} \mathcal{O} \cdots U_2 \mathcal{O} U_1 \, \ket{0..0},
\ee
where $\mathcal{O}\in \{I \otimes \mathcal{O}', \mathcal{O}' \otimes I: \mathcal{O}' \in\{ \mathcal{O}_A, \mathcal{O}_A^{-1}, \mathcal{O}_R, \mathcal{O}_R^{-1}, \mathcal{O}_C, \mathcal{O}_C^{-1}\}\}$, and $U_1,U_2, \ldots,U_T$ are unitaries independent of the oracles. 

\item[The block-encoding model:] 
For a matrix $A$, we assume that we are given a block-encoding of $A$, i.e., we are given the following unitary for some $\alpha\geq \|A\|$
\be \label{block-encoding}
W = \begin{pmatrix}
    A/\alpha & \cdot \\
    \cdot & \cdot 
\end{pmatrix}.
\ee
In this setting, we focus on quantum algorithms of the following form:
\be
\label{quantum alg2}
U_T \widetilde{W} U_{T-1} \widetilde{W} \cdots U_2 \widetilde{W} U_1 \ket{0..0},
\ee
where $\widetilde{W} \in \{I\otimes W', W' \otimes I: W'\in \{W, W^\dag, c\text{-}W, c\text{-}W^\dag\}\}$ and $U_1,U_2,\ldots,U_T$ are unitaries independent of $W$. Here, $c\text{-}W$ means control-$W$.
\end{description}

Let $\ket{\psi_A}$ be the final state of (\ref{quantum alg1}) or (\ref{quantum alg2}). Regarding the output of quantum algorithms, if the goal is to produce a quantum state, then we assume that $\ket{\psi_A} = \sqrt{1-\varepsilon} \, \ket{0} \ket{\phi_0} + \sqrt{\varepsilon} \, \ket{1} \ket{\phi_1}$ for some $\varepsilon\in[0,1/3]$. There is flag qubit (i.e., $\ket{0}$) that tells us the desired state (i.e., $\ket{\phi_0}$). If the output state is a number, say a bit, then we assume that $\ket{\psi_A} = p \ket{0} \ket{\phi_0} + \sqrt{1-p^2} \ket{1} \ket{\phi_1}$. If the output is 0 then $|p|^2 \geq 2/3$, otherwise $|p|^2 \leq 1/3$. For the above two tasks, we only need to measure the first qubit. With high probability, we will obtain the desired result.

There are some connections between the above two models.
To see this, recall that the block-encoding of a sparse matrix is constructed as follows \cite{gilyen2019quantum}: assume that $\max_{i,j} |A_{ij}|\leq 1$ for simplicity, define unitary
\be
\label{sparse BE 1}
U_L:\ket{0} \ket{i}  \rightarrow \ket{i}\ket{0} 
\rightarrow \frac{1}{\sqrt{s}}\sum_{j=1}^s \ket{i} \ket{j} 
\xrightarrow{\mathcal{O}_R}
\frac{1}{\sqrt{s}}\sum_{j=1}^s \ket{i} \ket{r_{ij}} .
\ee
And define unitary
\bea
U_R: && \ket{0} \ket{j} \ket{0}\ket{0}
\rightarrow \frac{1}{\sqrt{s}}\sum_{i=1}^s \ket{i} \ket{j} \ket{0}\ket{0}
\xrightarrow{\mathcal{O}_C \otimes I \otimes I}
\frac{1}{\sqrt{s}}\sum_{i=1}^s \ket{c_{ij}} \ket{j} \ket{0}\ket{0} \nonumber \\
&& \xrightarrow{\mathcal{O}_A\otimes I}
\frac{1}{\sqrt{s}}\sum_{i=1}^s \ket{c_{ij}} \ket{j} \ket{A_{c_{ij},j}} \ket{0} \nonumber  \\
&& \rightarrow  \frac{1}{\sqrt{s}}\sum_{i=1}^s \ket{c_{ij}} \ket{j} \ket{A_{c_{ij},j}} \left(A_{c_{ij},j} \ket{0} + \sqrt{1 - A_{c_{ij},j}^2} \ket{1} \right) \nonumber  \\
&& \xrightarrow{\mathcal{O}_A^{-1}\otimes I}  \frac{1}{\sqrt{s}}\sum_{i=1}^s \ket{c_{ij}} \ket{j} \ket{0} \left(A_{c_{ij},j} \ket{0} + \sqrt{1 - A_{c_{ij},j}^2} \ket{1} \right).
\label{sparse BE 2}
\eea
Then 
\beas
W &=&  (U_L\otimes I\otimes I)^\dag U_R
 \\
&=& U_1 (\mathcal{O}_R^{-1} \otimes I\otimes I) (\mathcal{O}_A^{-1}\otimes I) U_2 (\mathcal{O}_A \otimes I) (\mathcal{O}_C\otimes I\otimes I) U_3 \\
&=& \begin{pmatrix}
    A/s & \cdot \\
    \cdot & \cdot 
\end{pmatrix}
\eeas
is a block encoding of $A$, where $U_1, U_2, U_3$ are some unitaries independent of the oracles. As a result, if using this $W$, then the algorithm described in (\ref{quantum alg2}) is in the query model.

\subsection{Known results in these models}

It is possible that the oracles $\mathcal{O}_A,\mathcal{O}_R,\mathcal{O}_C$ or the block-encoding $W$ can be used in some ways that are very different from (\ref{quantum alg1}), (\ref{quantum alg2}) in a quantum algorithm, but most quantum algorithms for matrix problems known so far are in these two models. We below present some examples to support this.

\subsubsection{Quantum query algorithms for Boolean functions}

Let $f$ be a Boolean function. Given oracle access to $(x_1,\ldots,x_n) \in \{-1,1\}^n$, the goal is to compute $f(x_1,\ldots,x_n)$. A commonly used oracle is $\mathcal{O}: \ket{i} \ket{b} \mapsto \ket{i} \ket{b\oplus x_i}$, and a quantum query algorithm for evaluating $f(x_1,\ldots,x_n)$ works as the form $U_T \mathcal{O} U_{T-1} \mathcal{O} \cdots U_2 \mathcal{O} U_1 \, \ket{0}$. This fits into the model described in (\ref{quantum alg1}). Here we only need to view $(x_1,\ldots,x_n)$ as the diagonal matrix $\text{diag}(x_1,\ldots,x_n)$.

\subsubsection{Quantum eigenvalue transformation}

Let $A$ be a $p$-qubit operator, $\alpha, \varepsilon\in \mathbb{R}^+$ and $q\in \mathbb{N}$, then we say that the $(p+q)$-qubit unitary $W$ is an $(\alpha,q,\varepsilon)$-block-encoding of $A$, if $\|A - \alpha (\bra{0}^p\otimes I ) W (\ket{0}^p\otimes I )\| \leq \varepsilon.$
Namely, up to some error $\varepsilon$, $W$ has the form (\ref{block-encoding}). 

Block-encoding is a building block of the quantum singular value transformation \cite{gilyen2019quantum}. We below state it for Hermitian matrices. Let $P(x)$ be a polynomial of degree $d$. By \cite[Theorems 17 and 56]{gilyen2019quantum}, we can construct a $(1, q+2, 4d\sqrt{\varepsilon/\alpha})$ block encoding $\widetilde{W}$ of $P(A/\alpha)$ as follows.

\begin{itemize}
\item If $P(x)$ is an even polynomial, then there exist unitaries $U_1,\ldots,U_d$ at most depending on $P(x)$ such that
\be
\widetilde{W} = \prod_{k=1}^{d/2} \Big (U_{2k-1} W^\dag U_{2k} W\Big).
\ee
\item If $P(x)$ is an odd polynomial, then there exist unitaries  $W_1,\ldots,W_d$ at most  depending on $P(x)$ such that
\be
\widetilde{W} = U_1 W \prod_{k=1}^{(d-1)/2} \Big(U_{2k} W^\dag U_{2k+1} W\Big).
\ee
\item Generally, denote $\ell=\lfloor d/2 \rfloor$. Then 
there exist unitaries $U_1, U_2, \ldots, U_{2\ell + 2} $ at most depending on $P(x)$ such that
\be
\widetilde{W} 
= U_1
\begin{pmatrix}
I   & 0\\
0 & W
\end{pmatrix} 
\left( \prod_{k=1}^{\ell} 
U_{2k}
(I \otimes W^\dag)
U_{2k+1}
(I \otimes W) \right) U_{2\ell+2}.
\ee
\end{itemize}

As a block-encoding, we have
\[
\widetilde{W} \approx \begin{pmatrix}
P(A/\alpha) & \cdot \\
\cdot & \cdot 
\end{pmatrix}.
\]
For any given state $\ket{\b} = V \ket{0\cdots 0}$, applying $\widetilde{W}$ to $\ket{0}\ket{\b}$ gives
$
\ket{\psi} = \ket{0}\otimes P(A/\alpha)\ket{\b} + \beta \ket{1} \otimes \ket{G},
$
where $\beta$ is the amplitude and $\ket{G}$ is some garbage state.
To obtain the state $\ket{P(A/\alpha) \b}$, we can apply amplitude amplification. A key operator that will be used many times to increase the success probability is the reflection $2 \ket{\psi} \bra{\psi} - I
= \widetilde{W} (2\ket{0}\ket{\b} \bra{0}\bra{\b} - I) \widetilde{W}^\dag$. 
Now it is not hard to see that the above quantum algorithm that prepares the state $\ket{P(A/\alpha) \b}$ lies in the model described in (\ref{quantum alg2}). The uniatries $U_1,  U_2, \ldots, U_T$ are independent of $W$ but may depend on unitary $V$ that prepares the state $\ket{\b}$.

\subsubsection{Sparse Hamiltonian simulation and its applications}

The oracles (\ref{oracleA}), (\ref{oracleR}), (\ref{oracleC}) for sparse matrices were used by Berry and Childs \cite{berry2012black} to do Hamiltonian simulation. Their quantum algorithm fits into the model (\ref{quantum alg1}).\footnote{Although not in the same model, the quantum algorithm for recommendation systems also has a similar structure~\cite{kerenidis}.} The key idea of their quantum algorithm is using quantum phase estimation to a discrete-time quantum walk described by a unitary $V = i S (2TT^\dag - I)$, where $S \ket{j,k}=\ket{k,j}$ and $T = \sum_j \ket{j} \ket{\phi_j} \bra{j}$. Here if the  Hamiltonian is $H=(H_{jk})$, then $\ket{\phi_j} = \sqrt{\frac{\varepsilon}{\Lambda_1}} \sum_k \sqrt{H_{jk}^*} \ket{k}\ket{0} + \sqrt{1-\frac{\varepsilon \sigma_j}{\Lambda_1}} \ket{\zeta_j} \ket{1}$, where $\Lambda_1$ is a known upper bound of $\|H\|_1$, $\sigma_j=\sum_k |H_{jk}|, \ket{\zeta_j}$ is some garbage state, and $\varepsilon\in(0,1]$ is a small parameter that ensures a lazy quantum walk. The state $\ket{\phi_j}$ can be prepared using $O(1)$ calls to the oracles $\mathcal{O}_A,\mathcal{O}_R,\mathcal{O}_C$, which is similar to the process (\ref{sparse BE 2}). This is indeed how the oracles are used in their whole quantum algorithm for Hamiltonian simulation. So it fits into the model (\ref{quantum alg1}).
Hamiltonian simulation is a useful subroutine of many quantum algorithms for linear algebra problems. A famous example is the HHL algorithm for linear systems of equations \cite{harrow2009quantum}. As a result, all these algorithms found so far based on \cite{berry2012black} also fit into the model (\ref{quantum alg1}).

\subsection{Connections to quantum communication complexity}

Similar to the classical setting, we can also build connections between quantum protocols and quantum query algorithms for matrix problems. In the quantum case, we will mainly focus on the case that $k=2$, i.e., the Alice-Bob model. It turns out that this is enough for us. 
One key reason is that in the quantum case, the cost will be roughly $kT$ when there are $k$ players. This is in sharp contrast to $k+T$ in the classical case (see Theorem \ref{thm:connection1}). 
So we focus on the following simplified setting: Suppose Alice has a matrix $A^{(1)}\in \mathbb{R}^{m_1\times n}$ and a vector $ \b^{(1)} \in \mathbb{R}^{m_1}$, Bob has another matrix $A^{(2)}\in \mathbb{R}^{m_2\times n}$ and another vector $ \b^{(2)} \in \mathbb{R}^{m_2}$. Let $f(x)$ be a univariate function. Their goal is for Alice or Bob to output the state $\ket{f(A) \b}$, where
\be
A=\begin{pmatrix}
A^{(1)} \\
A^{(2)} \\
\end{pmatrix}, \quad
\b=\begin{pmatrix}
\b^{(1)} \\
\b^{(2)}
\end{pmatrix}.
\label{matrix}
\ee

\begin{prop}
\label{thm:model 1}
Assume that $A\in \mathbb{C}^{m\times n}$ has the decomposition (\ref{matrix}). 
For Alice or Bob to use the oracle $\mathcal{O}_A, \mathcal{O}_R$ or $\mathcal{O}_C$ once, they need to communicate $O(\log (mn))$ qubits.
\end{prop}

\begin{proof}
Let $\sum_{i,j} \alpha_{ij} \ket{i,j} \ket{0}$ be a quantum state in Alice's hand. To implement $\mathcal{O}_A$, she first applies $\mathcal{O}_{A^{(1)}}$ in her hand to the state to obtain $\sum_{i\leq m_1} \sum_j \alpha_{ij} \ket{i,j} \ket{A_{ij}} + \sum_{i> m_1} \sum_j \alpha_{ij} \ket{i,j} \ket{0}$. She then sends the state to Bob who applies $\mathcal{O}_{A^{(2)}}$ to it to obtain $\sum_{i,j} \alpha_{ij} \ket{i,j} \ket{A_{ij}}$. Bob then sends the state back to Alice. This costs $2\log(mn)$ qubits of communication. The protocol for using $\mathcal{O}_R$ is similar.

Next, we consider how to implement $\mathcal{O}_C$. 
For $i\in\{1,2\}$, denote the number of nonzero entries of $j$-th column of $A^{(i)}$ as $s_j^{(i)}$. So Alice knows $\{s_1^{(1)}, \ldots, s_n^{(1)}\}$ and Bob knows  $\{s_1^{(2)}, \ldots, s_n^{(2)}\}$. For any state $\sum_{i,j} \alpha_{ij} \ket{i} \ket{j} $ in Alice's hand, she first 
queries the indices in superposition as follows $\sum_{i,j} \alpha_{ij} \ket{i} \ket{j} \ket{s_j^{(1)}} = 
\sum_{i\leq s_j^{(1)}} \alpha_{ij} \ket{i} \ket{j} \ket{s_j^{(1)}}
+\sum_{i>s_j^{(1)}} \alpha_{ij} \ket{i} \ket{j} \ket{s_j^{(1)}} $. 
Next, she applies 
$\mathcal{O}_{C}$ for $A^{(1)}$ to this state  to obtain $\sum_{i\leq s_j^{(1)}} \alpha_{ij} \ket{c_{ij}} \ket{j} \ket{s_j^{(1)}}
+\sum_{i>s_j^{(1)}} \alpha_{ij} \ket{i+m_1} \ket{j} \ket{s_j^{(1)}} $.
Alice then sends the state to Bob and asks him to do similar operations for the second term. Here for the second term, Bob has to query the $(i-s_j^{(1)})$-th nonzero entry based on the oracle he has.
In the end, he sends the state back to Alice who undoes the query of $s_j^{(1)}$. This uses $2(2\log m+\log n)$ qubits of communication in total.
\end{proof}

From the above proof, if there are $k$-players and a coordinator, then for the coordinator to use the oracles once, they need $O(k\log(mn))$ qubits of communication.
The following result follows from step 1 of \cite[Proof of Proposition~14]{montanaro-shao2023}.

\begin{prop}
For each $i\in\{1,2\}$, let $W_i$ be an $(\alpha_i,1,0)$ block-encoding of $A^{(i)} \in \mathbb{R}^{m_i\times n}$. Then there is a quantum protocol for Alice or Bob to use an $(\sqrt{\alpha_1^2+\alpha_2^2},2,0)$ block-encoding of $A$ once with quantum communication complexity $O(\log(mn))$.
\end{prop}

With the above two propositions, it is not hard to obtain the following result. 

\begin{thm}
\label{thm:model 2}
If there is a quantum algorithm of the form (\ref{quantum alg1}) or (\ref{quantum alg2}) of complexity $O(T)$, then there is a quantum protocol for the same task with quantum communication complexity $O(T \log (mn))$.
\end{thm}

\subsection{Applications: Lower bounds analysis of quantum query algorithms}

In this part, we present some applications of Theorem \ref{thm:model 2}. We will focus on a further simplified case than the setting (\ref{matrix}). Here suppose Alice has a matrix $A$ and Bob has a state $\ket{\b}$, their goal is to prepare the state $\ket{f(A)\b}$ by Alice or Bob, where $f$ is a public function. In the model (\ref{quantum alg1}) or (\ref{quantum alg2}), we now assume that the unitaries $U_1,U_2, \ldots, U_T$ can depend on the unitary that prepares the state $\ket{\b}$. We below use Theorem \ref{thm:model 2} to prove certain lower bounds in the quantum query model, either the number of queries or the number of block-encodings we must use. We can even obtain lower bounds with respect to the number of copies of $\ket{\b}$.
To be more exact, if it is Alice to output $\ket{f(A)\b}$, then there is no need to communicate the oracles or the block-encoding of $A$. What she needs is $U_1,U_2,\ldots,U_T$, which may depend on $\ket{\b}$. So we will obtain a lower bound on the number of copies of $\ket{\b}$. 
If Bob is the one to output $\ket{f(A)\b}$, then we will obtain a lower bound on the number of calls of the block-encoding of $A$, or the number of calls of queries to $A$.



\begin{prop}[Lower bounds for matrix inversion]
\label{prop:Lower bounds for matrix inversion}
Any quantum algorithm in the form (\ref{quantum alg1}) or (\ref{quantum alg2}) that prepares $\ket{A^+\b}$ requires making $\Omega(\kappa/\log (mn))$ calls to the oracles or block-encoding of $A$, and requires $\Omega(\kappa/\log (mn))$ copies of the state $\ket{\b}$.
\end{prop}

\begin{proof}
The proof is inspired by \cite[Theorem 10 (3)]{montanaro2022quantum}. Let Alice have $(a_1,\ldots,a_n)\in \{0,1\}^n$ and Bob have $(b_1,\ldots,b_n)\in \{0,1\}^n$. Alice constructs a diagonal matrix $A$ by setting $A_{ii}=1/\sqrt{n}$ if $a_i=1$ and $1$ if $a_i=0$. Bob constructs a quantum state $\ket{\b} \propto \sum_{b_i=1} \ket{i}$. Then the quantum state $\ket{A^{-1}\b} \propto \sqrt{n} \sum_{a_i=b_i=1} \ket{i} + \sum_{a_i=0,b_i=1} \ket{i} $. When there is only one intersection, then we will set index $i$ with $a_i=b_i=1$ with probability close to $1/2$. So if we have the state $\ket{A^{-1}\b}$, we then can solve the Set-Disjointness problem using constantly many invocations of the above, whose communication complexity is $\Omega(\sqrt{n})=\Omega(\kappa)$. By Theorem \ref{thm:model 2}, we obtain the claimed lower bound.
\end{proof}

{\color{black} Apart from the $\log(mn)$ factor, the lower bound obtained above coincides with known lower bound for matrix inversion \cite{harrow2009quantum}.}
Note that in the query model, there is a similar proof based on Grover's algorithm \cite{grover1996fast}. Let $(x_1,\ldots,x_n)\in\{0,1\}^n$. Assume that there is an $i_0$ such that $x_{i_0}=1$ and $x_j=0$ for all $j\neq i_0$. It is known that the quantum query complexity of finding $i_0$ is $\Theta(\sqrt{n})$.
Let $A$ be a diagonal matrix with $A_{ii}=1/\sqrt{n}$ if $i=i_0$ and $A_{ii}=1$ if $i\neq i_0$. Let $\ket{\b}= \frac{1}{\sqrt{n}} \sum_i \ket{i}$. Then the quantum state $\ket{A^{-1}\b} = \frac{1}{\sqrt{2n-1}}(\sqrt{n} \ket{i_0} + \sum_{j\neq i_0} \ket{j})$. If we can prepare this state, then we can solve the search problem by measuring $\ket{A^{-1}\b}$. For this problem, the condition number $\kappa=\sqrt{n}$. So it gives a lower bound $\Omega(\kappa)$ of query complexity for the matrix inversion problem. Although the lower bound in Proposition \ref{prop:Lower bounds for matrix inversion} is weaker by a factor of $\log (mn)$, it gives a lower bound on the number of copies of the state $\ket{\b}$ at the same time. Moreover, the lower bounds also hold in the block-encoding model.

\begin{prop}[Lower bounds for matrix exponential]
Any quantum algorithm in the form (\ref{quantum alg1}) or (\ref{quantum alg2}) that prepares $\ket{e^{At}\b}$ requires making $\Omega(e^t/\log (mn))$ calls to the oracles or block-encoding of $A$, and requires $\Omega(e^t/\log (mn))$ copies of the state $\ket{\b}$.
\end{prop}

\begin{proof}
The proof is similar to the previous one. Imagine that Alice has $\a\in\{0,1\}^n$ and Bob has $\b\in\{0,1\}^n$. Alice constructs a diagonal matrix $A$ using $\a$. Then 
\[
e^{At} \b = \sum_{i=1}^n e^{a_i t} b_i \ket{i}
= \sum_{a_i=b_i=1} e^{t} \ket{i} + \sum_{a_i=0,b_i=1} \ket{i}.
\]
Let $t= \frac{1}{2} \ln n$. If we can prepare $\ket{e^{At} \b}$, then sampling from it returns $i$ such that $a_i=b_i=1$ with probability close to $1/2$. So, just like in the previous proof, we can solve the Set-Disjointness problem using this routine constantly many times. As the lower bound of the Set-Disjointness problem is $\Omega(\sqrt{n})$ and $e^t=\sqrt{n}$, by Theorem \ref{thm:model 2}, we obtain a lower bound of $\Omega(e^{t}/\log(mn))$ for preparing the state $\ket{e^{At}\b}$. 
\end{proof}

Recall that in \cite[Proposition 14]{an2022theory}, An et al.~showed that the number of copies of $\ket{\b}$ required to preparing state $\ket{e^{At}\b}$ is roughly $\Omega(e^{t})$. In the above, we gave a simple proof of this and also proved the same lower bound on the query complexity with respect to $A$ at the same time.

\begin{prop}[Lower bounds for matrix powers]
Any quantum algorithm in the form (\ref{quantum alg1}) or (\ref{quantum alg2}) that prepares $\ket{A^d\b}$ requires making $\Omega((\kappa^d+1/\gamma)/\log(mn))$ calls to the oracles or block-encoding of $A$, and requires $\Omega((\kappa^d+1/\gamma)/\log(mn))$ copies of the state $\ket{\b}$.
\end{prop}

\begin{proof}
We will use similar constructions as the above proofs. Alice constructs a diagonal matrix $A$ with diagonal entry $A_{ii}=1$ if $a_i=1$ and $A_{ii}=\varepsilon$ if $a_i=0$. Bob uses this input to define the state $\ket{\b}$. Then the state $\ket{A^d \b}$ is proportional to $\sum_{a_i=b_i=1}\ket{i} + \varepsilon^d \sum_{a_i=0, b_i=1}\ket{i}$. We will choose $\varepsilon$ such that $\varepsilon^{2d} n = 1$. So we can solve the Set-Disjointness problem with high probability by measuring this state. The condition number of $A$ is $\kappa=1/\varepsilon$ and so $\kappa^d = \sqrt{n}$. It equals the complexity of the Set-Disjointness problem. So we obtain a lower bound of $\Omega(\kappa^d)$.
To prove $\Omega(1/\gamma)$, Alice sets $\varepsilon=0$. Now $\kappa=1$ and $\gamma = \|A^d\b\|/\|\b\| = 1/\sqrt{n}$. 
\end{proof}

{\color{black} In \cite{montanaro-shao2023}, the authors studied the lower bound of quantum query complexity of computing $f(A)_{i,j}$, where $A$ is sparse and $f$ is a continuous function. They proved a low bound in terms of the approximate degree of $f(x)$. If $f(x)=x^d$, then the approximate degree is $\Theta(\sqrt{d})$ and so the lower bound of computing $A^d_{i,j}$ is $\Omega(\sqrt{d})$. However, this task is weaker than preparing the state $\ket{A^d\b}$. If we apply quantum singular value transform \cite{gilyen2019quantum}, then given the block-encoding of $A$, we can construct a block-encoding of $A^d$. To sample from $\ket{A^d\b}$, we need to do the postselection. The success probability depends on $\|A^d\b\|\geq 1/\kappa^d$ if assuming $\|A\|=1$. This leads to a quantum algorithm of complexity $O(\kappa^d)$, which almost matches our lower bound.}

\section{Conclusions and open problems}

In this work, we built a connection between quantum-inspired classical algorithms and communication complexity. In particular, we showed that quantum-inspired classical algorithms can be efficiently simulated using communication protocols. Using this connection, we proved some lower bounds of quantum-inspired classical algorithms for some matrix-relevant problems. We also generalised the idea to the quantum case. 

The following are some open problems that we think merit further study.

\begin{enumerate}
    \item The most exciting problem is to prove better lower bounds for the tasks mentioned in this work. We feel some new ideas are required. For example, the known results on communication complexity used in this paper are Set-Disjointness and Gap-Hamming problem in the coordinator model. If we have more candidate hard problems for communication in this model, we may use them to obtain better lower bounds. It is also possible that more efficient quantum-inspired classical algorithms exist such that the lower bounds obtained in this paper are close to tight. So the exploration of more efficient quantum-inspired classical algorithms is also very important.
    
    \item The quantum query complexity of functions of matrices is widely studied in the past \cite{gilyen2019quantum,montanaro-shao2023}, see the references in \cite{montanaro-shao2023} for a summary. In this work, we provided a new method to prove lower bounds of quantum query complexity, so it is interesting to see if some nontrivial new lower bounds can be obtained.
    
    \item In this paper, we only used the connection on one side, i.e., proving lower bounds. As discussed at the end of Section \ref{section:the connection}, via the connection, we can also propose efficient protocols for many matrix-relevant problems from quantum-inspired classical algorithms we know so far, e.g., \cite{bakshi2023improved,shao2022faster,chia2022sampling}. So it is interesting to see some useful applications of this method in communication complexity. Theorem \ref{thm:connection2} might be helpful in this direction.
\end{enumerate}

\section*{Acknowledgements}

This research is supported by the National Key Research Project of China under Grant No. 2023YFA1009403.

\appendix

\section{Another connection}
\label{section:another connection}

Apart from the connection discussed in Section \ref{section:the connection}, another way to see the connection, which should be more general in practice, is considering $A=\sum_{i=1}^k \lambda_i A^{(i)}$ and $\b=\sum_{i=1}^k \mu_i \b^{(i)}$. 
First, we list some definitions.

\begin{defn}
Let $p=(p_1,\ldots,p_n),q=(q_1,\ldots,q_n) \in \mathbb{R}^n_{\geq 0}$ be two distributions, i.e., $\sum p_i=\sum q_i=1$ we say $p$ $\phi$-oversamples $q$ if for all $i\in [n]$, $p_i \geq q_i/\phi$. 
\end{defn}

\begin{defn}[Oversampling and query access to a vector]
\label{defn:oversampling}
Let $\v=(v_1,\ldots,v_n) \in \mathbb{C}^n$ and $\phi\geq 1$, we have $SQ_{\phi}(\v)$, $\phi$-oversampling and query access to $\v$, if we have $Q(\v)$ and $SQ(\tilde{\v})$ for some vector $\tilde{\v} =(\tilde{v}_1,\ldots,\tilde{v}_n) \in \mathbb{C}^n$ satisfying $\|\tilde{\v}\|^2 = \phi \|\v\|$ and $|\tilde{v}_i|^2 \geq |v_i|^2$ for all $i\in [n]$.
\end{defn}

\begin{defn}[Oversampling and query access to a matrix]
Let $A=(A_{ij}) \in \mathbb{C}^{m\times n}$ and $\phi\geq 1$, we have $SQ_{\phi}(\v)$, $\phi$-oversampling and query access to $\v$, if we have $Q(A)$ and $SQ(\widetilde{A})$ for some matrix $\widetilde{A}=(\widetilde{A}_{ij}) \in \mathbb{C}^{m\times n}$ satisfying $\|\widetilde{A}\|_{\F}^2 = \phi \|A\|_{\F}^2$ and $|\widetilde{A}_{ij}|^2 \geq |A_{ij}|^2$ for all $(i,j)\in [m]\times [n]$.
\end{defn}

\begin{lem}[Lemma 3.5 of \cite{chia2022sampling}]
\label{lem3.5}
Suppose we are given $SQ_{\phi}(\v)$ and $\delta\in(0,1]$, then we can sample from $D_{\v}$ with probability $\geq 1-\delta$ with $O(\phi \log(1/\delta))$ applications of $SQ_{\phi}(\v)$. We can also estimate $\|\v\|$ to $\varepsilon$ multiplicative error for $\varepsilon \in (0,1]$ with probability  $\geq 1-\delta$ with $O(\varepsilon^{-2})$ applications of $SQ_{\phi}(\v)$.
\end{lem}

\begin{thm}
\label{thm:connection2}
Assume that player $\P_i$ holds a matrix $A^{(i)}\in \mathbb{R}^{m\times n}$ and a vector $\b^{(i)}\in \mathbb{R}^{m}$ whose entries are specified by $O(\log q)$ bits, where $i\in [k]$. Let $A=\sum_i \lambda_i A^{(i)}$ and $\b=\sum_i \mu_i \b^{(i)}$. Then in the coordinator model, we have the following.

\begin{itemize}

\item For the coordinator to use $SQ(\b)$ $O(T)$ times, it costs $O(kT\phi_b\log (qm))$ bits of communication, where
\[
\phi_b = k \frac{\sum_i  \|\mu_i \b^{(i)}\|^2 }{\|\b \|^2}.
\]

\item For the coordinator to use $SQ(A)$ $O(T)$ times, it costs $O(kT\phi_A \log (qmn))$ bits of communication, where
\[
\phi_A = k \frac{\sum_i  \|\lambda_i A^{(i)}\|_\F^2}{\|A\|_\F^2}.
\]

\end{itemize}
\end{thm}

In comparison, Theorem \ref{thm:connection1} has a lower complexity than the one stated in Theorem \ref{thm:connection2}. This is the main reason why we choose to use Theorem \ref{thm:connection1} to prove lower bounds.
However, the case discussed in Theorem \ref{thm:connection2} seems more powerful since, for example, the linear regression $\argmin\|A\x-\b\|$ with $A,\b$ given in (\ref{thm:eq}) is equivalent to $\argmin_\x\|\sum_i (A^{(i)})^\T A^{(i)} \x - \sum_i (A^{(i)})^\T \b^{(i)}\|$, which is a special case of the one discussed above. We feel Theorem \ref{thm:connection2} can play more important roles in designing efficient classical protocols for problems, such as machine learning problems or Hamiltonian simulation.

\begin{proof}
First, we show how the coordinator uses $SQ(\b)$.
We will use $b^{(i)}_j$ and $b_j$ to denote the $j$-th entry of $\b^{(i)}$ and $\b$ respectively.
Given $SQ(\b^{(i)})$ for all $i\in [k]$, we can construct $SQ_{\phi_b}(\b)$ (see Definition \ref{defn:oversampling} for $SQ_{\phi_b}$) in the following way, which is based on the proof of \cite[Lemma 3.6]{chia2022sampling} and also provides a protocol for the coordinator to use $SQ(\b)$.

\begin{enumerate}
\item Obtain $Q(\b)$: For each $j\in[m]$, we can compute $b_j=\sum_{i=1}^k \mu_i b^{(i)}_j$ by querying all $b^{(i)}_j$.

\item Obtain $SQ(\tilde{\b})$: Define $\tilde{\b}$ via
\[
\tilde{b}_j = \sqrt{k \sum_{i=1}^k |\mu_i b^{(i)}_j|^2}, \quad j \in [m].
\]
It is easy to check that $\|\tilde{\b}\|^2 = \phi_b \|\b\|^2$ and $|\tilde{b}_j| \geq |b_j|$ for all $j\in [m]$.
\begin{enumerate}
\item Query for $\|\tilde{\b}\|^2$: We can compute $\|\tilde{\b}\|^2 = k \sum_i  \|\mu_i \b^{(i)}\|^2$ directly by querying  $\|\b^{(i)}\|^2$ for all $i\in[k]$.
\item Query for $\tilde{b}_j$: We can compute $\tilde{b}_j$ by querying $ b^{(i)}_j$ for all $i\in[k]$.
\item Sample from $\tilde{\b}$: We sample an index $i\in[k]$ with probability $\|\mu_i \b^{(i)}\|^2/\sum_j\|\mu_j \b^{(j)}\|^2$, then take a sample $j\in[m]$ from the distribution defined by $\title{\b}^{(i)}$. We can check that the probability is
\[
\sum_{i=1}^k \frac{\|\mu_i \b^{(i)}\|^2}{\sum_j\|\mu_j \b^{(j)}\|^2} \frac{|b^{(i)}_j|^2}{\|\b^{(i)}\|^2}
= \frac{\sum_{i=1}^k |\mu_i b^{(i)}_j|^2}{\sum_j\|\mu_j \b^{(j)}\|^2}
= \frac{|\tilde{b}_j|^2}{\|\title{\b}\|^2}.
\]
\end{enumerate}
\end{enumerate}

The above procedure provides a protocol for the coordinator to use $SQ_{\phi_b}(\b)$, from which the coordinator can construct $SQ(\b)$. Firstly, querying for entries of $\b$ is straightforward, which costs $O(k \log q)$ bits of communication. Secondly, query for $\|\tilde{\b}\|^2$ and $\tilde{b}_j$ are also straightforward, which costs $O(k \log q)$ bits of communication. Thirdly, to sample from $\tilde{\b}$, player $\P_i$ first sends $\|\tilde{\b}^{(i)}\|$ to the coordinator. This costs $O(k \log (qm))$ bits of communication. Now the coordinator can do step (2.c) discussed in the above procedure with $O(\log q)$ bits of communication with player $\P_i$. In summary, for the coordinator to use $SQ_{\phi_b}(\b)$ $O(T)$ times, they use $O(T\log q+kT\log (qm))=O(kT\log (qm))$ bits of communication. Combining Lemma \ref{lem3.5}, it costs $O(kT\phi_b \log (qm))$ bits of communication for the coordinator to use $SQ(\b)$ $O(T)$ times.

Similar results also hold for linear combinations of matrices. 
Given $SQ(A^{(i)})$ for all $i\in [k]$, we can construct $SQ_{\phi}(A)$ in the following way, which is based on the proof of \cite[Lemma 3.9]{chia2022sampling}. Below for any matrix $M$, we use $M_{i,j}$ to denote its $(i,j)$-th entry and use $M_{i*}$ to denote the $i$-th row.

\begin{enumerate}
\item Obtain $Q(A)$: For each $(i,j)\in[m]\times [n]$, we can compute $A_{i,j}=\sum_{t=1}^k \lambda_t A^{(t)}_{i,j}$ by querying all $A^{(t)}_{i,j}$.

\item Obtain $SQ(\tilde{A})$: Define $\tilde{A}$ via
\[
\tilde{A}_{i,j} = \sqrt{k \sum_{t=1}^k |\lambda_t A^{(t)}_{i,j}|^2},
\]
It is easy to check that $\|\tilde{A}\|_\F^2 = \phi_A \|A\|_\F^2, |\tilde{A}_{ij}| \geq |A_{ij}|$ for all $(i,j)\in[m]\times [n]$.
\begin{enumerate}
\item Query for $\|\tilde{A}\|_\F^2$: We can compute $\|\tilde{A}\|_\F^2 = \sum_{t=1}^k  \|\lambda_t A^{(t)}\|_\F^2$ directly by querying  $\| A^{(i)}\|_\F^2$ for all $i\in[k]$.
\item Query for $\tilde{A}_{i,j}$: We can compute $\tilde{A}_{i,j}$ by querying $A^{(t)}_{i,j}$ for all $t\in[k]$.
\item Sample from $i$-th row $\tilde{A}_{i*}$: We first sample an index $s\in[k]$ with probability 
$\frac{\|\lambda_s A^{(s)}_{i*}\|^2}{\sum_{t=1}^k \|\lambda_t A^{(t)}_{i*}\|^2},$
then take a sample $j\in[m]$ from $A^{(s)}_{i*}$. The probability of seeing $j$ is
\[
\sum_{s=1}^k\frac{\|\lambda_s A^{(s)}_{i*}\|^2}{\sum_{t=1}^k \|\lambda_t A^{(t)}_{i*}\|^2} \frac{|A^{(s)}_{i,j}|^2}{\|A^{(s)}_{i*}\|^2}
= \frac{\sum_{s=1}^k|\lambda_s A^{(s)}_{i,j}|^2}{\sum_{t=1}^k \|\lambda_t A^{(t)}_{i*}\|^2}
= \frac{|\widetilde{A}_{ij}|^2}{\|\widetilde{A}_{i*}\|^2}.
\]
\item Obtain $SQ(\a)$ where $\tilde{\a}=(\|\tilde{A}_{1*}\|, \ldots, \|\tilde{A}_{m*}\|)$: Note that $\|\tilde{A}_{i*}\|^2 = k \sum_t \|\lambda_t \tilde{A}^{(t)} _{i*}\|^2$. So we can compute it by querying $\|\tilde{A}^{(t)} _{i*}\|^2$ for all $t\in[k]$. To sample from $\tilde{\a}$, we first sample $s\in[k]$ with probability ${\|\lambda_s A^{(s)}\|_\F^2}/{\sum_t \|\lambda_t A^{(t)}\|_\F^2}$, then take a sample $i\in[m]$ from the distribution defined by $\a^{(s)} = (\|A^{(s)}_{1*}\|,\ldots,\|A^{(s)}_{m*}\|)$. The probability is
\[
\sum_s\frac{\|\lambda_s A^{(s)}\|_\F^2}{\sum_t \|\lambda_t A^{(t)}\|_\F^2}
\frac{\|A^{(s)}_{i*}\|^2}{\|A^{(s)}\|_\F^2}
=\frac{\sum_s\|\lambda_s A^{(s)}_{i*}\|^2}{\sum_t \|\lambda_t A^{(t)}\|_\F^2}
=\frac{\|\tilde{A}_{i*}\|^2}{\|\tilde{A}\|_\F^2}.
\]
\end{enumerate}
\end{enumerate}

The analysis for the coordinator to use $SQ_{\phi}(A)$ is similar. For the coordinator to query $A_{i,j}, \tilde{A}_{i,j}$, it costs $O(k \log q)$ bits of communication. To query $\|\tilde{A}\|_\F^2=k\sum_i \|\lambda_i A^{(i)}\|_\F^2$, it costs $O(k\log (qmn))$ bits of communication. Here $O(\log (qmn))$ is the number of bits used to specify $\|A^{(i)}\|_\F^2$. To query for $\|\tilde{A}_{i*}\|^2 = k \sum_t \|\lambda_t A^{(t)} _{i*}\|^2$, it costs $O(k \log (qn))$ bits of communication. To sample from $\widetilde{A}_{i*}$, each player $\P_i$ first sends $\|A^{(s)}_{i*}\|$ to the coordinator, this uses $O(k(\log q+\log n))$ bits of communication. To sample from $\tilde{\a}$, since the coordinator already has the information of $\|A^{(i)}\|_\F^2$ for all $i\in[k]$, this step now only costs $O(\log q)$ bits of communication (namely communicate with player $\P_i$ once for some $i$). 
Finally, to sample from $\tilde{A}_{i*}$, each player needs to send $\|A^{(s)}_{i*}\|$ to the coordinator, then the coordinator has the distribution of seeing $s$ with probability ${\|\lambda_s A^{(s)}_{i*}\|^2}/{\sum_t \|\lambda_t A^{(t)}_{i*}\|^2},$ this uses $O(\log (qn))$ bits of communication. Thus for the coordinator to use $SQ_{\phi_A}(A)$ $O(T)$ times, it costs $O(T\log (qn)+kT\log (qmn))$ bits of communication. By Lemma \ref{lem3.5}, it costs $O(kT\phi_A \log (qmn))$ bits of communication for the coordinator to use $SQ(A)$ $O(T)$ times.
\end{proof}

\bibliography{quantumview-template}

\begin{thebibliography}{10}

\bibitem{tang2019quantum}
Ewin Tang.
\newblock ``A quantum-inspired classical algorithm for recommendation systems''.
\newblock In Proceedings of the 51st annual ACM SIGACT Symposium on Theory of Computing.
\newblock \href{https://dx.doi.org/https://doi.org/10.1145/3313276.3316310}{Pages 217--228}.
\newblock ~(2019).

\bibitem{chia2022sampling}
Nai-Hui Chia, Andr{\'a}s~Pal Gily{\'e}n, Tongyang Li, Han-Hsuan Lin, Ewin Tang, and Chunhao Wang.
\newblock ``Sampling-based sublinear low-rank matrix arithmetic framework for dequantizing quantum machine learning''.
\newblock \href{https://dx.doi.org/https://doi.org/10.1145/3549524}{Journal of the ACM {\bf 69}, 1--72}~(2022).

\bibitem{chakraborty_et_al:LIPIcs:2019:10609}
Shantanav Chakraborty, Andr{\'a}s Gily{\'e}n, and Stacey Jeffery.
\newblock ``{The Power of Block-Encoded Matrix Powers: Improved Regression Techniques via Faster Hamiltonian Simulation}''.
\newblock In 46th International Colloquium on Automata, Languages, and Programming (ICALP 2019).
\newblock \href{https://dx.doi.org/https://doi.org/10.4230/LIPIcs.ICALP.2019.33}{Volume 132, pages 33:1--33:14}.
\newblock ~(2019).

\bibitem{bakshi2023improved}
Ainesh Bakshi and Ewin Tang.
\newblock ``An improved classical singular value transformation for quantum machine learning''.
\newblock In Proceedings of the 2024 Annual ACM-SIAM Symposium on Discrete Algorithms (SODA).
\newblock \href{https://dx.doi.org/https://doi.org/10.1137/1.9781611977912.86}{Pages 2398--2453}.
\newblock SIAM~(2024).

\bibitem{shao2022faster}
Changpeng Shao and Ashley Montanaro.
\newblock ``Faster quantum-inspired algorithms for solving linear systems''.
\newblock \href{https://dx.doi.org/https://doi.org/10.1145/3520141}{ACM Transactions on Quantum Computing {\bf 3}, 1--23}~(2022).

\bibitem{gilyen2022improved}
Andr{\'a}s Gily{\'e}n, Zhao Song, and Ewin Tang.
\newblock ``An improved quantum-inspired algorithm for linear regression''.
\newblock \href{https://dx.doi.org/https://doi.org/10.22331/q-2022-06-30-754}{Quantum {\bf 6}, 754}~(2022).

\bibitem{lloyd2013quantum}
Seth Lloyd, Masoud Mohseni, and Patrick Rebentrost.
\newblock ``Quantum algorithms for supervised and unsupervised machine learning''~(2013).
\newblock  \href{http://arxiv.org/abs/1307.0411}{arXiv:1307.0411}.

\bibitem{kerenidis2017quantum}
Iordanis Kerenidis and Anupam Prakash.
\newblock ``Quantum recommendation systems''.
\newblock In 8th Innovations in Theoretical Computer Science Conference (ITCS 2017).
\newblock \href{https://dx.doi.org/https://doi.org/10.4230/LIPIcs.ITCS.2017.49}{Page 49:1–49:21}.
\newblock ~(2017).

\bibitem{gilyen2019quantum}
Andr{\'a}s Gily{\'e}n, Yuan Su, Guang~Hao Low, and Nathan Wiebe.
\newblock ``Quantum singular value transformation and beyond: exponential improvements for quantum matrix arithmetics''.
\newblock In Proceedings of the 51st Annual ACM SIGACT Symposium on Theory of Computing.
\newblock \href{https://dx.doi.org/https://doi.org/10.1145/3313276.3316366}{Pages 193--204}.
\newblock ~(2019).

\bibitem{lloyd2014quantum-pca}
Seth Lloyd, Masoud Mohseni, and Patrick Rebentrost.
\newblock ``Quantum principal component analysis''.
\newblock \href{https://dx.doi.org/https://doi.org/10.1038/nphys3029}{Nature Physics {\bf 10}, 631--633}~(2014).

\bibitem{harrow2009quantum}
Aram~W Harrow, Avinatan Hassidim, and Seth Lloyd.
\newblock ``Quantum algorithm for linear systems of equations''.
\newblock \href{https://dx.doi.org/https://doi.org/10.1103/PhysRevLett.103.150502}{Physical Review Letters {\bf 103}, 150502}~(2009).

\bibitem{tang2021quantum}
Ewin Tang.
\newblock ``Quantum principal component analysis only achieves an exponential speedup because of its state preparation assumptions''.
\newblock \href{https://dx.doi.org/https://doi.org/10.1103/PhysRevLett.127.060503}{Physical Review Letters {\bf 127}, 060503}~(2021).

\bibitem{gilyen2018quantum}
Andr{\'a}s Gily{\'e}n, Seth Lloyd, and Ewin Tang.
\newblock ``Quantum-inspired low-rank stochastic regression with logarithmic dependence on the dimension''~(2018).
\newblock  \href{http://arxiv.org/abs/1811.04909}{arXiv:1811.04909}.

\bibitem{gharibian2022dequantizing}
Sevag Gharibian and Fran{\c{c}}ois Le~Gall.
\newblock ``Dequantizing the quantum singular value transformation: hardness and applications to quantum chemistry and the quantum pcp conjecture''.
\newblock In Proceedings of the 54th Annual ACM SIGACT Symposium on Theory of Computing.
\newblock \href{https://dx.doi.org/https://doi.org/10.1145/3519935.3519991}{Pages 19--32}.
\newblock ~(2022).

\bibitem{montanaro-shao2023}
Ashley Montanaro and Changpeng Shao.
\newblock ``Quantum and classical query complexities of functions of matrices''.
\newblock In Proceedings of the 56th Annual ACM Symposium on Theory of Computing.
\newblock \href{https://dx.doi.org/https://doi.org/10.1145/3618260.364966}{Pages 573--584}.
\newblock ~(2024).

\bibitem{montanaro2022quantum}
Ashley Montanaro and Changpeng Shao.
\newblock ``Quantum communication complexity of linear regression''.
\newblock \href{https://dx.doi.org/https://doi.org/10.1145/3625225}{ACM Transactions on Computation Theory {\bf 16}, 1--30}~(2024).

\bibitem{yao1979some}
Andrew Chi-Chih Yao.
\newblock ``Some complexity questions related to distributive computing (preliminary report)''.
\newblock In Proceedings of the 11th annual ACM Symposium on Theory of Computing.
\newblock \href{https://dx.doi.org/https://doi.org/10.1145/800135.804414}{Pages 209--213}.
\newblock ~(1979).

\bibitem{de2002quantum}
Ronald~de Wolf.
\newblock ``Quantum communication and complexity''.
\newblock \href{https://dx.doi.org/https://doi.org/10.1016/S0304-3975(02)00377-8}{Theoretical Computer Science {\bf 287}, 337--353}~(2002).

\bibitem{rao2020communication}
Anup Rao and Amir Yehudayoff.
\newblock ``Communication complexity: and applications''.
\newblock Cambridge University Press. Cambridge, U.K.~(2020).

\bibitem{phillips2016lower}
Jeff~M Phillips, Elad Verbin, and Qin Zhang.
\newblock ``Lower bounds for number-in-hand multiparty communication complexity, made easy''.
\newblock In Proceedings of the twenty-third annual ACM-SIAM Symposium on Discrete Algorithms (SODA).
\newblock \href{https://dx.doi.org/https://doi.org/10.1137/15M1007525}{Pages 486--501}.
\newblock SIAM~(2012).

\bibitem{kalyanasundaram1992probabilistic}
Bala Kalyanasundaram and Georg Schintger.
\newblock ``The probabilistic communication complexity of set intersection''.
\newblock \href{https://dx.doi.org/https://doi.org/10.1137/0405044}{SIAM Journal on Discrete Mathematics {\bf 5}, 545--557}~(1992).

\bibitem{chakrabarti2011optimal}
Amit Chakrabarti and Oded Regev.
\newblock ``An optimal lower bound on the communication complexity of gap-hamming-distance''.
\newblock In Proceedings of the forty-third annual ACM Symposium on Theory of Computing.
\newblock \href{https://dx.doi.org/https://doi.org/10.1145/1993636.1993644}{Pages 51--60}.
\newblock ~(2011).

\bibitem{woodruff2017distributed}
David~P Woodruff and Qin Zhang.
\newblock ``When distributed computation is communication expensive''.
\newblock \href{https://dx.doi.org/https://doi.org/10.1007/978-3-642-41527-2_2}{Distributed Computing {\bf 30}, 309--323}~(2017).

\bibitem{li2023ell_p}
Yi~Li, Honghao Lin, and David Woodruff.
\newblock ``$\ell_p$-regression in the arbitrary partition model of communication''.
\newblock In Gergely Neu and Lorenzo Rosasco, editors, Proceedings of Thirty Sixth Conference on Learning Theory.
\newblock Volume 195 of Proceedings of Machine Learning Research, pages 4902--4928.
\newblock PMLR~(2023).
\newblock  url:~\url{https://proceedings.mlr.press/v195/li23b.html}.

\bibitem{montanaro2019quantum}
Ashley Montanaro.
\newblock ``Quantum states cannot be transmitted efficiently classically''.
\newblock \href{https://dx.doi.org/https://doi.org/10.22331/q-2019-06-28-154}{Quantum {\bf 3}, 154}~(2019).

\bibitem{nielsen2010quantum}
Michael~A Nielsen and Isaac~L Chuang.
\newblock ``{Quantum Computation and Quantum Information}''.
\newblock Cambridge University Press. Cambridge, U.K.~(2010).

\bibitem{berry2012black}
Dominic~W Berry and Andrew~M Childs.
\newblock ``{Black-box Hamiltonian simulation and unitary implementation}''.
\newblock \href{https://dx.doi.org/https://doi.org/10.26421/QIC12.1-2-4}{Quantum Information and Computation {\bf 12}, 29--62}~(2012).

\bibitem{kerenidis}
Iordanis Kerenidis and Anupam Prakash.
\newblock ``{Quantum Recommendation Systems}''.
\newblock In Christos~H. Papadimitriou, editor, 8th Innovations in Theoretical Computer Science Conference (ITCS 2017).
\newblock \href{https://dx.doi.org/https://doi.org/10.4230/LIPIcs.ITCS.2017.49}{Volume~67 of Leibniz International Proceedings in Informatics (LIPIcs), pages 49:1--49:21}.
\newblock Dagstuhl, Germany~(2017). Schloss Dagstuhl--Leibniz-Zentrum fuer Informatik.

\bibitem{grover1996fast}
Lov~K Grover.
\newblock ``A fast quantum mechanical algorithm for database search''.
\newblock In Proceedings of the twenty-eighth annual ACM Symposium on Theory of computing.
\newblock \href{https://dx.doi.org/https://doi.org/10.1145/237814.237866}{Pages 212--219}.
\newblock ~(1996).

\bibitem{an2022theory}
Dong An, Jin-Peng Liu, Daochen Wang, and Qi~Zhao.
\newblock ``A theory of quantum differential equation solvers: limitations and fast-forwarding''~(2022).
\newblock  \href{http://arxiv.org/abs/2211.05246}{arXiv:2211.05246}.

\end{thebibliography}
\bibliographystyle{quantum}

\end{document}